\documentclass[preprint,12pt]{elsarticle}
\usepackage{graphicx}
\usepackage{epsfig}
\usepackage{mathptmx}
\usepackage{times}
\usepackage{amsmath} 
\usepackage{amsthm}
\usepackage{amssymb}
\usepackage{mathtools}
\usepackage{tcolorbox}
\usepackage[ruled,vlined]{algorithm2e}
\usepackage[colorlinks = true, linkcolor = blue, citecolor =
blue]{hyperref}
\usepackage{enumitem}
\usepackage{subcaption}
\usepackage[mathscr]{euscript}
\usepackage{xcolor}

\usepackage{natbib}
\usepackage{geometry}
\usepackage{fleqn}
\usepackage{xspace}

\newcommand\new[1]{{\color{black}#1}}

\newtheorem{thm}{Theorem}
\newtheorem{prop}[thm]{Proposition}
\newtheorem{lem}[thm]{Lemma}
\newtheorem{cor}[thm]{Corollary}
\newtheorem{rem}[thm]{Remark}

\newcommand{\ldef}{:=}
\newcommand{\rdef}{=:}

\newcommand{\real}{\ensuremath{\mathbb{R}}}

\newcommand{\realp}{\ensuremath{\mathbb{R}_{>0}}}
\newcommand{\realz}{\ensuremath{\mathbb{R}_{\ge 0}}}
\newcommand{\nzreal}{\mathbb{R}\setminus\{0\}}
\newcommand{\nzrealn}{\mathbb{R}^n\setminus\{0\}}
\newcommand{\nzrealtwo}{\mathbb{R}^2\setminus\{0\}}
\newcommand{\ints}{{\mathbb{Z}}}
\newcommand{\nat}{{\mathbb{N}}}
\newcommand{\natz}{{\mathbb{N}}_0}
\newcommand{\e}{\mathrm{e}}
\newcommand{\dd}{\mathrm{d}}
\newcommand{\cl}[1]{\mathrm{cl}(#1)}
\newcommand{\norm}[1]{\left\lVert #1 \right\rVert}

\newcommand{\tr}{\mathrm{tr}}
\newcommand{\tmin}{\tau_{\min}}
\newcommand{\tmax}{\tau_{\max}}
\newcommand{\tavg}{\tau_{\text{avg}}}

\newcommand{\thmtitle}[1]{\mbox{}\textit{(#1).}}

\newcommand{\remend}{\relax\ifmmode\else\unskip\hfill\fi\hbox{$\bullet$}}

\begin{document}
	
\title{Asymptotic Behavior of Inter-Event Times in
	Planar Systems under Event-Triggered Control}
\author[1]{Anusree Rajan\corref{cor1}}
\ead{anusreerajan@iisc.ac.in}

\author[1,2]{Pavankumar Tallapragada}
\ead{pavant@iisc.ac.in}

\cortext[cor1]{Corresponding author}

\affiliation[1]{organization={Department of Electrical Engineering, Indian Institute of Science, Bengaluru}
}
\affiliation[2]{organization={Robert Bosch Centre for Cyber Physical Systems, Indian Institute of Science, Bengaluru}
	}

%

\begin{abstract}
  This paper analyzes the asymptotic behavior of inter-event times in
  planar linear systems, under event-triggered control with a general
  class of scale-invariant event triggering rules. In this setting,
  the inter-event time is a function of the ``angle'' of the state at
  an event. This viewpoint allows us to analyze the inter-event times by
  studying the fixed points of the \emph{angle map}, which represents the
  evolution of the ``angle'' of the state from one event to the next.
  We provide a sufficient condition for the 
  convergence or non-convergence of inter-event times to a steady state 
  value under a scale-invariant event-triggering rule. Following up on 
  this, we further analyze the inter-event time behavior in the special 
  case of threshold based event-triggering rule and we provide various 
  conditions for convergence or non-convergence of inter-event times to 
  a constant. We also analyze
  the asymptotic average inter-event time as a function of the
  angle of the initial state of the system. With the help of ergodic 
  theory, we provide a sufficient condition
  for the asymptotic average inter-event time to be a constant for all
  non-zero initial states of the system. Then, we consider a special
  case where the \emph{angle map} is an orientation-preserving
  homeomorphism. Using rotation theory, we comment on the asymptotic
  behavior of the inter-event times, including on whether the
  inter-event times converge to a periodic sequence.  We illustrate the
  proposed results through numerical simulations.
\end{abstract}



\maketitle

\section{Introduction}

Event-triggered control is commonly used in several applications with
resource constraints. Efficiency of this control method is due to the
state dependent and non-constant inter-event times, which are
implicitly determined by a state-dependent event-triggering
rule. However, this also means that the evolution of the inter-event
times is difficult to predict, which makes higher level planning and
scheduling difficult. Further, there is not enough work that
analytically quantifies the improvement in resource usage by
event-triggered controllers compared to time-triggered
controllers. From both these points of view, it is very useful to
analyze the inter-event times generated by event-triggered
controllers. For example, understanding the evolution of inter-event
	times helps to schedule multiple processes over a shared
	communication channel or to plan transmissions under
	constraints. Similarly, understanding inter-event times generated by
	an event-triggering rule can help in the analytical
	quantification of the improvement of average inter-event times for
	an event-triggered controller over that of a time-triggered
	controller. With these motivations, in this paper, we carry out a
systematic analysis of the asymptotic behavior of inter-event times
for planar linear systems under a general class of scale-invariant
event-triggering rules.

\subsection{Literature review}

Event-triggered control is a popular control method in the field of
networked control systems~\cite{PT:2007, WH:2012, ML:2010,
  DT-SH:2017-book}. In the literature on this topic, inter-event time
analysis is typically limited to showing the existence of a positive
lower bound on the inter-event times to guarantee the absence of zeno
behavior. Among the exceptions to this rule, a few works provide
bounds on the average sampling rate~\cite{FB:2017, PT-MF-JC:2018-tac,
  SB-PT:2021-cta}.  References~\cite{KA-BB:2002,BD-AL-DQ:2017}
consider a scalar stochastic event-triggered control system and
provide a closed form expression for the expected average sampling
period or communication rate. There are also some works~\cite{PT:2016,
  QL:2017, JP-JPH-DL:2017, MJK-PT-JC-MF:2020-tac} that determine the
necessary and sufficient data rates for achieving the control
objective irrespective of the controller that is used.
References~\cite{BAK-DJA-WPMH:2018, FDB-DA-FA:2018} take a different
point of view and design event triggering rules that guarantee better
performance than periodic time-triggered control, for a given average
sampling rate. Reference~\cite{PT-MF-JC:2019-ijsc} designs an
event-triggered controller that ensures exponential stability of the
closed loop system while satisfying some given interval constraints on
event times. Whereas, self-triggered~\cite{AA:2010} and periodic event
triggered~\cite{WH:2013} control methods guarantee the absence of zeno
behavior by design.

Evolution of inter-event times is far less studied topic in the
literature. We believe reference~\cite{MV-PM-EB:2009} is the first
paper to analyze the periodic and chaotic patterns exhibited by the
inter-event sequences of linear time invariant systems under
homogeneous event-triggering rules.  Reference~\cite{RP-RS-WH-2022}
analyzes the evolution of inter-event times for planar linear systems
under time-regularized relative thresholding event-triggering
rule. Specifically, this paper explains the commonly observed
behaviors of inter-event times, such as steady-state convergence and
oscillatory nature, under the ``small'' thresholding parameter and
``small'' time-regularization parameter scenario.  However, these
results are qualitative in nature as they do not clearly specify the
bounds on the parameters for which the claims hold. At the same time,
~\cite{RP-RS-WH-2022} does not provide explicit bounds on the behavior
of inter-event times. Reference~\cite{AK-MM:2018} takes a different
approach to characterize the sampling behavior of linear
time-invariant event-triggered control systems by using finite-state
abstractions of the system. The same idea is extended to nonlinear and
stochastic event-triggered control systems by
references~\cite{GD-MM:2020} and \cite{GD-LL-MM:2021},
respectively. On the other hand,~\cite{GAG-MM:2021} provides a
framework to estimate the smallest, over all initial states, average
inter-sample time of a linear periodic event-triggered control system
by using finite-state abstractions. Reference~\cite{GG-MM:2021}
improves the above approach by showing robustness to small enough
model uncertainties. The paper~\cite{GG-MM:2022} shows that the
abstraction based method can also be used to analyze the chaotic
behavior exhibited by the traffic patterns of periodic event-triggered
control systems.

Our previous work~\cite{AR-PT:2020} analyzes the evolution of
inter-event times for planar linear systems under a general class of
event-triggering rules. This work is a continuation of the same. In one of our recent works~\cite{AR-PT:2023}, we analyze the inter-event time evolution in linear systems under region-based self-triggered control. In this control
method, the state space is partitioned into a finite number of conic
regions and each region is associated with a fixed inter-event
time.

\subsection{Contributions}

The major contribution of our work is that we analyze the asymptotic behavior of the inter-event times, such as convergence to a constant or to a periodic sequence, in planar linear systems under a general class of scale-invariant event-triggering rules. We carry out this analysis by essentially studying
how the ``angle'' of the state, the angle of the state in polar coordinates, evolves from one event to the next. We also leverage
the literature on ergodic theory and rotation theory in our
analysis. Under mild technical assumptions, we provide a mathematical explanation for different kinds of asymptotic behavior of the ``angle'' of the state and as a
consequence the asymptotic behavior of the inter-event times. We also analyze the asymptotic average inter-event
time as a function of the ``angle'' of the initial state of the
system. Our results are quantitative in nature and are applicable
for a very broad class of event-triggering rules. 

Note that, analyzing the evolution of inter-event times is complex 
even for planar systems. The results in the paper are among very few in 
the literature that seek to explain the variety of evolutions that
	is possible for the inter-event times.
  Thus inter-event time analysis even for planar linear systems is 
  useful for building intuition and ideas for more complicated systems.
  For example, the idea that analyzing the state evolution from one 
  event to next as a means to analyzing the evolution of inter-event 
  times does certainly apply to $n$-dimensional systems. We use the 
  same idea in our recent paper~\cite{AR-PT:2023}, where we analyze the 
  inter-event time behavior for linear systems under region based 
  self-triggered control. We next provide an overview of the 
  contributions of our paper relative to closely related works from the 
  literature.

While~\cite{MV-PM-EB:2009} seeks to study the evolution of
  inter-event times by understanding the state evolution, the results
  in the paper are quite preliminary. In the current paper, we provide
  several necessary conditions and sufficient conditions on the system parameters
  which could be used to predict convergence or lack of convergence of
  inter-event times to a constant. The results
in~\cite{RP-RS-WH-2022} are restricted to planar linear systems under 
time-regularized relative thresholding based event-triggering in the 
``small relative threshold
parameter and time-regularization parameter'' setting. 
\cite{RP-RS-WH-2022} does not explicitly mention the bounds on these 
parameters for which the results hold, nor is the derivation of such 
bounds obvious from the analysis. In contrast, our results hold for all 
range of parameters. In fact, aided
by one of our analytical results, we show through simulations that the
analytical results and interpretations in~\cite{RP-RS-WH-2022} are not
necessarily true for large enough parameters. Further, 
~\cite{MV-PM-EB:2009,
  RP-RS-WH-2022} analyze specific triggering rules, whereas our
results hold for a general class of scale invariant
rules. References~\cite{AK-MM:2018, GD-MM:2020, GD-LL-MM:2021,
  GAG-MM:2021, GG-MM:2021, GG-MM:2022} characterize the sampling
behavior of event-triggered control systems by using finite
state-space abstractions of the system. However, this approach can be
computationally very demanding.
   
The main contributions of this paper with respect to our previous
work~\cite{AR-PT:2020} are stability analysis of the
fixed points of the \emph{angle map} under any scale-invariant quadratic
event-triggering rule and a framework to analyze the asymptotic average 
inter-event time as a function of the initial state of the system. We
also provide several important new results and improve some of the
existing results. 

\subsection{Organization}

Section~\ref{sec:problem-setup} formally sets up the problem and
states the objective of this paper. Section~\ref{sec:inter-event-time}
and Section~\ref{sec:evol-inter-event} analyze the properties of the
inter-event time as a function of the state at an event and the
steady-state behavior of inter-event times under the event-triggered
control method, respectively. In
Section~\ref{sec:Asymptotic-avg-inter-event-time}, with the help of
ergodic theory and rotation theory, we study the asymptotic average
inter-event time as a function of the initial state of the
system. Section~\ref{sec:numerical-examples} illustrates the results
using numerical examples. Finally, we provide some concluding remarks
in Section~\ref{sec:conclusion}.

\subsection{Notation}

Let $\real$, $\realz$, and $\realp$ denote the set of all real,
non-negative real and positive real numbers, respectively. $\nzreal$
and $\nzrealn$ denote the set of all non-zero real numbers and the set
of all non-zero vectors in $\real^n$, respectively. Let $\nat$ and
$\natz$ denote the set of all positive and non-negative
integers, respectively. For any $x \in \real^n$, $\norm{x}$ denotes the euclidean norm of $x$. For an $n \times n$ square matrix $A$, let
$\det(A)$ and $\tr(A)$ denote determinant and trace of $A$,
respectively. $B_{\epsilon}(u) \ldef \{x \in \real^n: \norm{x-u}\le \epsilon\}$ represents an $n-$dimensional ball of
radius $\epsilon$ centered at $u \in \real^n$. Let
$(X,\mathcal{B},\mu)$ be a measure space where $X$ is a set,
$\mathcal{B}=\mathcal{B}(X)$ is the borel $\sigma-$algebra on the set
$X$ and $\mu$ is a measure on the measurable space $(X,\mathcal{B})$.

\section{Problem Setup}\label{sec:problem-setup}

In this section, we formulate the problem of analyzing the asymptotic
behavior of inter-event times in event-triggered control systems. We
begin by specifying the class of systems and event-triggering rules
that we consider and then state the main objective of this paper.

\subsection{System Dynamics}

Consider a linear time invariant planar system,
\begin{subequations}\label{eq:system}
  \begin{equation}\label{eq:plant_dyn}
    \dot{x}(t) = Ax(t) + Bu(t),
  \end{equation}
  where \(x\in\real^2\) is the plant state and \(u\in\real^m\) is the
  control input, while $A \in \real^{2\times 2}$ and
  $B \in \real^{2\times m}$ are the system matrices. Consider a
  sampled data controller and let \(\{t_k\}_{k\in \natz}\) be the
  sequence of event times at which the state is sampled and the
  control input is updated as follows,
  \begin{equation}\label{eq:control_input}
    u(t)=Kx(t_k), \quad \forall t\in[t_k, t_{k+1}).
  \end{equation}
\end{subequations}
Let the control gain \(K\) be such that \(A_c \ldef A+BK\) is Hurwitz.

\subsection{Triggering Rules}

In this paper, we assume that the event times \(\{t_k\}_{k\in \natz}\)
are generated in an event-triggered manner so as
to implicitly guarantee asymptotic stability of the origin of the
closed loop system. It is common to construct such event-triggering
rules based on a candidate Lyapunov function. For example, consider a
quadratic candidate Lyapunov function $V(x) = x^TPx$, where
\(P \in \real^{2\times 2} \) is a positive definite symmetric matrix
that satisfies the Lyapunov equation
\begin{equation}
  P A_c + A_c^T P = - Q , \label{eq:lyap}
\end{equation}
for a given symmetric positive definite matrix $Q$. Following are
three different event-triggering rules that are commonly used in the
literature for stabilization tasks.
\begin{subequations}\label{eq:tr}
  \begin{align}
    t_{k+1} &= \min \{t>t_k:\dot{V}(x(t))=0\} \label{eq:tr1}\\
    t_{k+1} &= \min\{t>t_k:\norm{x(t_k)-x(t)}=\sigma
              \norm{x(t)}\}, \label{eq:tr2} \\
    t_{k+1} &= \min\{t>t_k:V(x(t)) =
              V(x(t_k))e^{-r(t-t_k)}\}. \label{eq:tr3}
  \end{align}
\end{subequations}
First two triggering rules render the origin of the closed loop system
asymptotically stable, with $\sigma$ sufficiently small in the latter
rule (see~\cite{PT:2007,WH:2012} for example). The third event-triggering rule ensures exponential stability
for a sufficiently small $r > 0$ (see~\cite{PT:2016} for example).

During the inter-event intervals, we can write the solution \(x(t)\) of system~\eqref{eq:system} as
\begin{equation*}
  x(t)=G(\tau)x(t_k), \quad  \forall t \in [t_k,t_{k+1}),
\end{equation*}
where \(\tau \ldef t - t_k\) and
\begin{equation*}
  G(\tau) \ldef \e^{A \tau} + \int_0^{\tau} \e^{A (\tau - s)}
  \dd s (A_c - A) .
\end{equation*}
Using this structure of the solution, we can write the three triggering rules~\eqref{eq:tr} as
\begin{equation}\label{eq:general_tr_rule}
  t_{k+1}-t_k=\min\{\tau>0:f(x(t_k),\tau):= x^T(t_k)M(\tau)x(t_k)=0\},
\end{equation}
where \(M(\tau)\) is a time varying symmetric matrix. In particular, for the triggering rules~\eqref{eq:tr1}-\eqref{eq:tr3} $M(\tau)$ is equal to
\(M_1(\tau)\), \(M_2(\tau)\) and \(M_3(\tau)\), respectively, where 
\begin{subequations}\label{eq:Ms}
  \begin{align}
    M_1(\tau) &\ldef \frac{\dd G^T( \tau )}{\dd \tau} PG (\tau) + G^T(
    \tau ) P \frac{ \dd G(\tau) }{ \dd \tau } \\
    M_2(\tau) &\ldef (1 - \sigma^2) G^T(\tau) G(\tau) - (G^T(\tau) +
                G(\tau)) + I \label{eq:M2}
    \\
    M_3(\tau) &\ldef G^T(\tau)PG(\tau)-P\e^{-r\tau} 
  \end{align}
\end{subequations}

Note that if $A$ is invertible, then the expression and computation of
$M(\tau)$ is simplified significantly as
\begin{equation*}
  G(\tau) = I + A^{-1}(\e^{A\tau}-I)A_c .
\end{equation*}

\subsection{Objective}

The main objective of this paper is to analyze the evolution of
inter-event times along the trajectories of system \eqref{eq:system}
for the general class of event triggering rules
\eqref{eq:general_tr_rule}. We seek to provide analytical guarantees
for the asymptotic behavior of inter-event times under these
rules. Specifically, we would like to answer the questions: when do
the inter-event times converge to a steady-state value or to a
periodic sequence and when do the asymptotic average inter-event times
becomes a constant for all initial states of the system. We also want
to analyze the asymptotic average inter-event time as a function of
the initial state of the system. The approach we take is to analyze
inter-event time and the state at the next event as functions of the
state at the time of the current event.

\section{Inter-event time as a function of the state}\label{sec:inter-event-time}

In this section, we analyze the inter-event time $t_{k+1}-t_k$, as a
function of $x(t_k)$, for the system~\eqref{eq:system} under the
general class of event triggering
rules~\eqref{eq:general_tr_rule}. Note that most of the results in 
this section were first proposed in our previous 
paper~\cite{AR-PT:2020}. However, this paper includes proofs of all the 
results proposed in~\cite{AR-PT:2020}.

Next, formally, we define the inter-event
time function \(\tau_e:\nzrealtwo \rightarrow \realp\) as
\begin{equation}\label{eq:inter_event_time}
  \tau_e(x) \ldef \min\{\tau>0:f(x,\tau) = x^T M(\tau) x = 0\}.
\end{equation}
We can write $t_{k+1}-t_k= \tau_e(x(t_k))$ for all $k \in \natz$.
Next, we analyze the properties of this inter-event time function such
as scale-invariance, periodicity and continuity. 

\subsection{Properties of the inter-event time function}

\begin{rem}\thmtitle{The inter-event time function is scale-invariant} \label{rem:scale_inv}
  Note from~\eqref{eq:inter_event_time} that
  $f(\alpha x, \tau) = \alpha^2 f(x, \tau)$ for all $\alpha \in \real$
  and $x \in \real^2$. Hence, \(\tau_e(\alpha x)=\tau_e(x)\), for any
  \(x \in \nzrealtwo\) and for any \(\alpha \in \nzreal\). \remend
\end{rem}

The scale-invariance property implies that we can redefine the
\emph{inter-event time function} for planar systems as a scalar function
\(\tau_s: \real \rightarrow \realp\),
\begin{equation}\label{eq:IET_theta}
  \tau_s(\theta) \ldef \min\{\tau>0:f_s(\theta,\tau) \ldef
  x_\theta^TM(\tau)x_\theta = 0\},
\end{equation}
where \(x_\theta \ldef \begin{bmatrix} \cos(\theta) & \sin(\theta)
\end{bmatrix}^T\), so that \(\tau_e(x)=\tau_s(\theta)\) for 
\(x = \alpha x_\theta\) for all  $\alpha \in \nzreal$. Hence for a
planar system, the inter-event time $t_{k+1}-t_k=\tau_s(\theta_k)$ for
all $k \in \natz$, where $\theta_k$ is the angle between $x(t_k)$ and
the $x_1$ axis. 

\begin{rem}\thmtitle{\(\tau_s(\theta)\) is a periodic function with
  period \(\pi\)} \label{rem:IET_periodic}
  We know that for \newline
  $x_\theta=\begin{bmatrix} \cos(\theta) &
  \sin(\theta) \end{bmatrix}^T$,
  $\tau_s(\theta) = \tau_e(x_{\theta}) = \tau_e(-x_{\theta}) =
  \tau_s(\theta + \pi)$ for all $\theta \in \real$. \remend
\end{rem}

Periodicity of \(\tau_s(\theta)\) helps us to restrict our analysis to
the domain \([0,\pi)\). Next, we present an important property of
$f_s(\theta, \tau)$ that plays a major role in the subsequent analysis.

\begin{lem}\thmtitle{For any fixed \(\tau\), \(f_s(\theta,\tau)\) is a
    sinusoidal function with a shift in phase and
  mean} \label{lem:f_with_fixed_tau}
  Let $m_{ij}(\tau)$ be the $(ij)^{th}$ element of
  $M(\tau) \in \real^{2 \times 2}$. For any fixed \(\tau \in \realp\),
  \begin{align} %
        &f_s(\theta,\tau) = \frac{ \tr(M(\tau)) }{ 2 } + a \sin \left( 2
        \theta + \arctan \left( b \right) \right), \label{eq:f_sin}\\
        & a \ldef \frac{1}{2} \sqrt{ (\tr(M(\tau)))^2 - 4 \det(M(\tau))}, \
        b \ldef \frac{ m_{11}(\tau) - m_{22}(\tau) }{ 2m_{12}(\tau) }
        . \notag
  \end{align}
\end{lem}

\begin{proof}
  Here, we skip the time argument of $m_{ij}(\tau)$ for brevity. Note
  that $m_{12} = m_{21}$. Then, for any fixed $\tau \in \realp$,
  \begin{align}
    f_s(\theta,\tau)%
     &=\begin{bmatrix}
       \cos(\theta) & \sin(\theta)
       \end{bmatrix} M(\tau) \begin{bmatrix} \cos(\theta) \\ \sin(\theta)
       \end{bmatrix}, \label{eqn:f-quadratic}
    \\
     &= m_{11} \cos^2( \theta ) + m_{22} \sin^2( \theta ) + 2m_{12}
       \cos( \theta )\sin( \theta ), \notag\\
     &=m_{11}+(m_{22}-m_{11})\sin^2(\theta)+m_{12}\sin(2\theta), \notag\\
     &=\frac{ m_{11} + m_{22} }{2} + \frac{ m_{11} - m_{22} }{2} \cos(
       2\theta ) + m_{12}\sin(2\theta), \notag
  \end{align}
  which when suitably re-expressed gives the result.
\end{proof}

Using the structure of $f_s(\theta, \tau)$ in~\eqref{eq:f_sin}
  and the quadratic form~\eqref{eqn:f-quadratic}, we can easily
determine the number of solutions to $f_s(\theta, \tau) = 0$ for any
fixed $\tau$.
\begin{cor}\thmtitle{Number of solutions $\theta$ to
    $f_s(\theta, \tau) = 0$ for a fixed
    $\tau$} \label{cor:f_num_roots}
  For any fixed \(\tau \in \realp\), if \(\det (M(\tau))>0\), then
  \(f_s(\theta,\tau)=0\) has no solutions; if \(\det (M(\tau))=0\)
  then \(f_s(\theta,\tau)=0\) has a single solution
  $\theta \in [0, \pi)$ or \(f_s(\theta,\tau)=0\) for all
  \(\theta \in [0,\pi)\); if $\det (M(\tau))<0$ then
  \(f_s(\theta,\tau)=0\) has exactly two solutions
  \(\theta \in [0,\pi)\). \qed
\end{cor}

\begin{proof}
	Note that, for any fixed \(\tau \in \realp\), $\det(M(\tau))>0$ implies $|\tr(M(\tau))|>|a|$ where $a$ is defined as in~\eqref{eq:f_sin}. This implies that the magnitude of the shift in the mean of the sinusoidal function in~\eqref{eq:f_sin} is strictly greater than the maximum magnitude of the sinusoidal function. Hence, according to Lemma~\ref{lem:f_with_fixed_tau}, we can say that \(f_s(\theta,\tau)=0\) has no solutions. Following similar arguments, we can say that if \(\det (M(\tau))=0\)
	then \(f_s(\theta,\tau)=0\) has a single solution
	$\theta \in [0, \pi)$ or \(f_s(\theta,\tau)=0\) for all
	\(\theta \in [0,\pi)\) if, additionally, $\tr(M(\tau))=0$. Similarly, if $\det (M(\tau))<0$ then
	\(f_s(\theta,\tau)=0\) has exactly two solutions
	\(\theta \in [0,\pi)\).
\end{proof}
From the quadratic form~\eqref{eqn:f-quadratic}, we can also
obtain a necessary and sufficient condition for the event-triggering
rule~\eqref{eq:general_tr_rule} to reduce to a periodic triggering
rule, with inter-event times that are independent of the state.

  \begin{cor}\thmtitle{Necessary and sufficient condition for the
      triggering rule~\eqref{eq:general_tr_rule} to reduce to periodic
    triggering} \label{cor:const-iet}
    \(\tau_s(\theta)= \tau_1, \forall \theta \in [0,\pi)\) if and only
    if $\det(M(\tau))>0$ for all $\tau \in (0,\tau_1)$,
    $\tau_1 = \min \{ \tau>0 : \det(M(\tau))=0 \}$ and $M(\tau_1) = 0$,
    the zero matrix. \qed
  \end{cor}

\begin{proof}
  First, we prove sufficiency. If $\det(M(\tau))>0$ for all $\tau \in 
  (0,\tau_1)$, then by Corollary~\ref{cor:f_num_roots} we know that for 
  each $\tau \in (0,\tau_1)$, $f_s(\theta, \tau) = 0$ has no solutions. 
  Hence, $\tau_s(\theta) \geq \tau_1$ for all $\theta$. If 
  additionally, $M(\tau_1) = 0$, the zero matrix, then from the 
  definition of $\tau_s(\theta)$ in~\eqref{eq:IET_theta} we can see 
  that $\tau_s(\theta) = \tau_1$ for all $\theta$.
  
	Now, we prove necessity. Let $\tau_{\min} = \min \{ \tau>0 : 
	\det(M(\tau)) \leq 0 \}$. Again from Corollary~\ref{cor:f_num_roots} 
	we know that there is no $\theta$ for which $\tau_s(\theta) < 
	\tau_{\min}$. Then, we know from Corollary~\ref{cor:f_num_roots} that 
	$\exists \theta$ such that $f_s(\theta, \tau_{\min}) = 0$ and hence 
	hence $\tau_s(\theta) = \tau_{\min}$. However, if $\det(M(\tau)) < 
	0$, then Corollary~\ref{cor:f_num_roots} implies that there are 
	exactly two values of $\theta$ in $[0, \pi)$ for which 
	$\tau_s(\theta) = \tau_{\min}$ and for all other $\theta$, 
	$\tau_s(\theta) > \tau_{\min}$. So, it must be that $\tau_{\min} = 
	\tau_1$ and $\det(M(\tau_{\min}))=0$. Finally, if $M(\tau_{\min}) = 
	M(\tau_1) \neq 0$, then again there is exactly one $\theta$ for which 
	$\tau_s(\theta) = \tau_{\min} = \tau_1$ and for all other $\theta$, 
	$\tau_s(\theta) > \tau_1$. This proves the necessity.
\end{proof}
 Note that this necessary and sufficient condition depends only
    on the time varying matrix $M(\tau)$, which can be determined
    given the system parameters and the event-triggering rule. If we
    know that the triggering rule for a given event-triggered control
    system is periodic, then further analysis of inter-event times is
    not required.
  
\subsection{Continuity of the inter-event time function
$\tau_s(\theta)$}

In this subsection, we seek to obtain conditions under which the
inter-event time function $\tau_s(\theta)$ is continuous. Towards this
aim, we make the following assumption about the matrix function
$M(\tau)$ since the general class of event-triggering
rules~\eqref{eq:general_tr_rule} for an arbitrary $M(\tau)$ is very
broad. 

\begin{enumerate}[resume, label=\textbf{(A\arabic*)},align=left]
  \item \hspace{.1em} Every element of the matrix $M(.)$ is a real analytic function
    of $\tau$ and there exists a $\tau_m$ such that \(M(\tau)\) is
    negative definite for $(0, \tau_m)$, where
    \begin{equation*}
      \tau_m \ldef \inf \{\tau >0 : \det(M(\tau))=0\}.
    \end{equation*}  \label{A:M}
\end{enumerate}
It is easy to verify that each $M_i(.)$ in~\eqref{eq:Ms},
corresponding to the three triggering rules \eqref{eq:tr}, satisfies
Assumption~\ref{A:M}. This is because in $M_1(.)$ and $M_2(.)$
  the dependence on $\tau$ comes from the matrix exponential
  $\mathrm{e}^{A \tau}$ and its integral with respect to $\tau$. In
  $M_3(.)$, there is an additional exponential function
  $\mathrm{e}^{-r \tau}$ which is combined linearly with other terms
  dependent on $\tau$. Letting $J := S^{-1} A S$ be the real Jordan
  form of $A$, we have
  $\mathrm{e}^{A \tau} = S \mathrm{e}^{J \tau} S^{-1}$. Thus, each
  element of $M_i(.)$ is a linear combination of products of
  exponential functions, polynomials (in case $A$ is not
  diagonalizable) and sinusoidal functions (in case $A$ has complex
  eigenvalues), all of which are real analytic functions of
  $\tau$. Thus, each of the $M_i(.)$'s are real analytic functions of
  $\tau$. Note that these arguments hold true even if $A$ is
  singular. Further, both $M_1(0)$ and $M_2(0)$ are negative
definite. Though $M_3(0) = 0$ the time derivative of $M_3$ at
$\tau = 0$, $\dot{M}_3(0)$, is negative definite for suitable $P$ and
$r$. Otherwise, the sequence of inter-event times generated by the 
event-triggering rule~\eqref{eq:tr3} would not have a positive lower 
bound on inter-event times.

Now, let \(\tmin\) and \(\tmax\) denote the global minimum and the
global maximum of \(\tau_s(\theta)\), respectively, that is,
\begin{equation*}
  \tmin \ldef \min_{\theta \in [0,\pi)} \tau_s(\theta), \quad 
  \tmax \ldef \max_{\theta \in [0,\pi)} \tau_s(\theta). 
\end{equation*}
For a matrix $M(.)$ that satisfies Assumption~\ref{A:M}, clearly $\tmin
= \tau_m$ as $\det(M(\tau))>0$ in the interval $(0, \tau_m)$ and
according to Corollary~\ref{cor:f_num_roots}, \(f_s(\theta,\tau)=0\) has
no solution for \(\tau \in (0,\tau_m)\) and has a solution for $\tau =
\tau_m$. In general, $\tmax$ may not exist, that is $\tmax = \infty$. In
this case, it means that there exists a $x_0 \in \nzrealtwo$ such that if
$x(t_k) = x_0$ then $t_{k+1} = \infty$. However, such an $x_0$ cannot
exist if $A$ has positive real parts for both its eigenvalues and if the
triggering rule~\eqref{eq:general_tr_rule} ensures $x = 0$ is
asymptotically stable. In such a case, $\tmax$ is a finite quantity.

We approach the question of continuity of the inter-event time
function $\tau_s(\theta)$ by first analyzing the smoothness properties
of the level set $f_s(\theta, \tau) = 0$ in the $(\theta, \tau)$
space. In particular, Assumption~\ref{A:M} implies that
$\det(M(\tau))$ is also real analytic and as a result it has finitely
many zeros in the interval $\tau \in [0, \tmax]$. This observation,
along with Corollary~\ref{cor:f_num_roots}, can be used to say that
$f_s(\theta, \tau) = 0$ has finitely many connected branches, which
are arbitrarily smooth, in the set
$\{ (\theta, \tau) \in [0, \pi) \times [0, \tmax] \}$.  We formally
state this claim in the following result.

\begin{lem}\thmtitle{The level set $f_s(\theta, \tau) = 0$ has
  finitely many connected branches, which are arbitrarily smooth}
  \label{lem:finite-branches}%
  Suppose that $M(.)$ in~\eqref{eq:IET_theta} satisfies
  Assumption~\ref{A:M} and $\tmax < \infty$. Then, the level set
  $f_s(\theta, \tau) = 0$ has finitely many connected branches in the
  set $\{ (\theta, \tau) \in [0, \pi) \times [0, \tmax] \}$. Each branch
  is an arbitrarily smooth curve in $(\theta, \tau)$ space and can be
  parameterized by $\tau$ in a closed interval.
\end{lem}

\begin{proof}
  First note that under Assumption~\ref{A:M}, all elements of $M(.)$
  are real analytic functions, which implies that $\det(M(\tau))$ is
  also a real analytic function of $\tau$. This is true because the
  determinant of a matrix is a polynomial of its elements, and
  products and sums of real analytic functions are also real
  analytic. As a consequence, on the closed and bounded interval
  $[0, \tmax]$, $\det(M(\tau))$ has finitely many zeros. This implies
  that there are finitely many sub-intervals $[g_i, h_i]$ of
  $[0, \tmax]$ such that $\det(M(g_i)) = \det(M(h_i)) = 0$ and
  $\det(M(\tau)) < 0$ for all $\tau \in (g_i, h_i)$. Then,
  Corollary~\ref{cor:f_num_roots} guarantees that
  $f_s(\theta, \tau) = 0$ has exactly two solutions for each
  $\tau \in [g_i, h_i]$ for each of the finitely many $i$, with the
  two solutions coincident at $g_i$ and $h_i$ but nowhere else. Thus,
  $f_s(\theta, \tau) = 0$ has finitely many branches in
  $\{ (\theta, \tau) \in [0, \pi) \times [0, \tmax] \}$. Smoothness of
  the branches is a consequence of the fact that $f_s(\theta, \tau)$
  is an arbitrarily smooth function, which is also evident
  from~\eqref{eq:f_sin}.
\end{proof}

Lemma~\ref{lem:finite-branches} allows us to apply the implicit function
theorem on $f_s(\theta, \tau) = 0$ at all $(\theta, \tau_s(\theta)) \in
[0, \pi) \times [0, \tmax]$, except at finitely many points. From this,
we guarantee that $\tau_s(\theta)$ is continuously differentiable in
$[0, \pi)$, except at finitely many points.

\begin{thm}\thmtitle{Inter-event time function is continuously
    differentiable except for finitely many
  $\theta$} \label{thm:ts-continuity} Suppose that $M(.)$
in~\eqref{eq:IET_theta} satisfies Assumption~\ref{A:M} and
$\tmax < \infty$. Then, the inter-event time function
$\tau_s(\theta)$, defined as in~\eqref{eq:IET_theta}, is continuously
differentiable on $[0, \pi)$ except at finitely many $\theta$.
\end{thm}

\begin{proof}
  Recalling Lemma~\ref{lem:finite-branches}, consider any one of the
  finitely many branches of the level set $f_s(\theta, \tau) = 0$ in
  the set $\{ (\theta, \tau) \in [0, \pi) \times [0, \tmax] \}$. We
  denote the smooth parameterization of the branch by $\tau$ as
  $\theta( \tau )$. Then, by Theorem~1 in~\cite{JS-VS:1972}
  (Morse-Sard Theorem for real analytic functions), we can say that
  the critical values $\theta$ of the function $\theta(\tau)$ form a
  finite set. We can infer two observations from this. First, in
  tracing out the $\tau_s(\theta)$ function, there are finitely many
  jumps between the branches of the level set $f_s(\theta, \tau) = 0$
  in the set $\{ (\theta, \tau) \in [0, \pi) \times [0, \tmax]
  \}$. Second, for all $\theta \in [0, \pi)$ except at finitely many
  $\theta$,
  $\frac{ \partial f_s( \theta, \tau ) }{ \partial \tau }|_{( \theta,
    \tau_s(\theta))} \neq 0$ and therefore the implicit function
  theorem guarantees continuous differentiability of $\tau_s(\theta)$
  on $[0, \pi)$ except at finitely many $\theta$.
\end{proof}

Based on Theorem~\ref{thm:ts-continuity} and its proof, we provide a
sufficient condition for $\tau_s(\theta)$ to be continuously
differentiable.

\begin{cor}\thmtitle{Corollary to
  Theorem~\ref{thm:ts-continuity}} \label{cor:suff-cond-ts-continuous}
  If
  $x_\theta^T\dot{M}(\tau)x_\theta \ne
  0$ for all $(\theta ,\tau) \in \real \times \real$
  such that $x_\theta^TM(\tau)x_\theta=0$ or if $\dot{M}(\tau) > 0$ for all $\tau \in [\tau_{\min},\tau_{\max}]$, then the inter-event
  time function $\tau_s:\real \rightarrow \realp$ defined as
  in~\eqref{eq:IET_theta} is continuously differentiable.  \qed
\end{cor}
Since in simulations, we encounter $\tau_s(\theta)$ functions that
\emph{visually seem to be continuous} quite often, we present the
following result in the special case where $\tau_s(\theta)$ is a
continuous function.

\begin{prop}\label{prop:global_extremum_tau_s}
  If the inter-event time function $\tau_s(\theta)$ is a continuous
  function, then every local extremum of \(\tau_s(\theta)\) is a global
  extremum.
\end{prop}

\begin{proof}
  We prove this result by contradiction. Suppose there exists an
  extremum of \(\tau_s(\theta)\) at $\theta_1$ with value
  \(\tau_1 \in (\tmin, \tmax)\). That is, the extremum at $\theta_1$
  is not a global extremum. Then, the assumptions that
  $\tau_s(\theta)$ is continuous and $\theta_1$ is a local extremizer
  and the fact that $\tau_s(\theta)$ is periodic with period $\pi$
  imply that there exist
  $\theta_2, \theta_3 \in (\theta_1, \theta_1 + \pi)$ such that
  $\tau_s(\theta_2) = \tau_s(\theta_3) = \tau_s(\theta_1) = \tau_1$.
  However, this contradicts Corollary~\ref{cor:f_num_roots}, which
  says that for any given $\tau_1 > 0$, $f_s(\theta, \tau_1) = 0$, and
  hence $\tau_s(\theta) = \tau_1$, can at most have two solutions for
  $\theta$.  Therefore, the claim in the result must be true.
\end{proof}

\section{Evolution of the inter-event time} \label{sec:evol-inter-event}

In this section, we provide a framework for analyzing the evolution of
the inter-event time along the trajectories of the system~\eqref{eq:system} under the general class of
event-triggering rules~\eqref{eq:general_tr_rule}. In the previous
section we showed that, for scale-invariant event-triggering rules,
the inter-event time is determined completely by the angle of the
state at the current event-triggering instant.
So, we restrict our analysis to the domain $R^1 \ldef \real / 2\pi \ints$, which is defined as
the quotient of real numbers by the equivalence relation of
differing by an integer multiple of $2\pi$. Then we define
a map $\phi:R^1 \to R^1$, referred to as \emph{angle map}, which
represents the evolution of the ``angle'' of the state from one event
to the next as, 

\begin{equation}\label{eq:Tmap}
  \theta_{k+1} = \phi(\theta_k) \ldef \text{arg} \left(
    G(\tau_s(\theta_k))%
    \begin{bmatrix}
      \cos(\theta_k) \\ \sin(\theta_k)
  \end{bmatrix} \right),
\end{equation}
where 
\begin{align*}
\text{arg}(x) \ldef 
\begin{cases}
\arctan( \frac{x_2}{x_1} ), \ &\text{if } x_1>0, x_2 \geq 0
\\
\pi + \arctan(\frac{x_2}{x_1}), \ &\text{if } x_1<0 \\
2\pi + \arctan(\frac{x_2}{x_1}), \ &\text{if } x_1>0, x_2<0 \\
\frac{\pi}{2},  \ &\text{if } x_1=0, x_2>0 \\
\frac{-\pi}{2},  \ &\text{if } x_1=0, x_2<0 \\
\text{undefined}, \ &\text{if } x_1=0, x_2=0.
\end{cases}
\end{align*}
and $\theta_k = \text{arg}(x(t_k))$ denotes the angle between the
state $x(t_k)$ and the positive $x_1$ axis. Thus, analysis of the
inter-event time function $\tau_s(\theta)$ and the \emph{\emph{angle map}}
$\phi(\theta)$ helps us to understand the evolution of the inter-event
time for an arbitrary initial condition $x(t_0)$. In particular, the
analysis of fixed points of the \emph{angle map} helps us to determine the
steady state behavior of the inter-event times. This is the main
  idea behind the results of this section.
  
  We first make an observation regarding the periodicity of $\phi(\theta)-\theta$ map and then we present the main results of this section.
  \begin{rem}\thmtitle{$\phi(\theta)-\theta$ is periodic with period $\pi$}\label{rem:phi-theta_periodic}
  	As the iner-event time function $\tau_s(\theta)$ is periodic with period $\pi$, $\phi(\theta+\pi)=\arg(G(\tau_s(\theta+\pi))x_{\theta+\pi})=\arg(-G(\tau_s(\theta))x_{\theta})=\phi(\theta)+\pi$. Thus $\phi(\theta+\pi)-(\theta+\pi)=\phi(\theta)-\theta$ for all $\theta \in \real$. \remend
  \end{rem}
\begin{rem}\thmtitle{Sufficient condition for the convergence of inter-event times to a steady state
		value} \label{rem:IET-analysis}
Suppose there exists a fixed point of the \emph{angle map}, i.e.,
	$\exists \theta$ s.t. $\phi(\theta)=\theta$. Then
	$t_{k+1} - t_k = \tau_s(\theta)$, $\forall k \in \natz$ and for
	all initial conditions $x(t_0) = \alpha
	\begin{bmatrix}
	\cos(\theta) & \sin(\theta)
	\end{bmatrix}^T$, with $\alpha \in \nzreal$.  Moreover, if
	$\theta$ is an asymptotically stable fixed point of the \emph{angle map} then
	$\displaystyle \lim_{k \to \infty}(t_{k+1} - t_k) =
	\tau_s(\theta)$ for all initial conditions in the region of
	convergence of $\theta$ under the \emph{angle map} $\phi(.)$.
\end{rem}
\begin{thm}\thmtitle{Sufficient condition for the non-convergence of inter-event times to a steady state
      value} \label{thm:IET-analysis} Consider the planar
    system~\eqref{eq:system} along with the event-triggering
    rule~\eqref{eq:general_tr_rule}, for a general $M(.)$ that
    satisfies Assumption~\ref{A:M}. If there does not exist $\theta \in [0,\pi)$ such that $\phi^k(\theta)-\theta=d\pi$ for some $d \in \ints$, $\forall k \in \{1,2\}$ and if the inter-event time function
    $\tau_s(.)$ is not a constant function, then the inter-event
    times do not converge to a steady state value for any initial
    state of the system.
    
  \end{thm}

  \begin{proof}		
   Note that the inter-event time
    converges to a constant $c$ if and only if there exists a subset
    of the level set $\tau_s(\theta)=c$ which is positively invariant
    under the \emph{angle map} $\phi(.)$. Note that the domain of $\phi(.)$
    is an interval of length $2\pi$. According to
    Corollary~\ref{cor:f_num_roots} and Remark~\ref{rem:IET_periodic},
    the level set $\tau_s(\theta) = c$ is either empty, or equal to
    $[0, 2\pi]$ or $\{\theta_1,\theta_1+\pi\}$ or $\{\theta_1,\theta_2,\theta_1+\pi,\theta_2+\pi\}$ for some $\theta_1,
   \theta_2 \in [0,\pi)$. Assuming $\tau_s(.)$ is not a constant function, there exists a subset
   of the level set $\tau_s(\theta)=c$ which is positively invariant
   under the angle map only if $\exists \theta \in [0,\pi)$ such that $\phi^k(\theta)-\theta=d\pi$ for some $d \in \ints$, for some $k \in \{1,2\}$. This completes the proof of this result.
  \end{proof}

Remark~\ref{rem:IET-analysis} and Theorem~\ref{thm:IET-analysis} establish a connection between
  the steady state behavior of inter-event times and the evolution of
  the angle under the \emph{angle map}. Having established this connection,
  in the rest of this section, we focus on analysis of the \emph{angle map}
  and its fixed points.
  
\subsection{Stability of the Fixed Points of the \emph{Angle map}}
  
Next, we are interested in analyzing the stability of the fixed points
of the \emph{angle map} as this will help us understand the steady state
behavior of the inter-event times. First, we make the following
  observation about the number of fixed points of the \emph{angle map}.
\begin{rem}\thmtitle{\emph{Angle map} has a bounded number of fixed points or
    every $\theta$ is a fixed point}\label{rem:finite_fps}
  Note that, there exists a fixed point for the $\phi(\theta)$ map if
  and only if there exists an $x \in \nzrealtwo$ such that $x(t_k)=x$
  implies $x(t_{k+1})=\alpha x$ for some $\alpha \in \nzreal$. This
  can happen if and only if $\det(L(\tau))=0$ for some
  $\tau \in \realp$, and $\exists x \in \real^2$ such that
  $\tau_e(x) = \tau$ and $L(\tau) x = 0$, where
  \begin{equation}\label{eq:L_tau} L(\tau) \ldef G(\tau) -
    \alpha I.
  \end{equation}
 As $\det(L(\tau))$ is an analytic function of $\tau$ under Assumption~\ref{A:M}, it has a
  bounded number of zeros in the interval $[\tmin, \tmax]$. So, if
  there does exist a $\tau \in [\tmin, \tmax]$ such that
  $\det(L(\tau)) = 0$ then either $\phi(\theta) = \theta$ for all
  $\theta \in [0, \pi)$ or the \emph{angle map} $\phi(.)$ has a bounded
  number of fixed points. \remend
\end{rem}

Next, we present a lemma which helps to prove the main result of this
subsection, which gives sufficient conditions for the stability and
instability of the fixed points of the \emph{angle map}.  Then, we make some
observations that are used for further analysis.
  
 \begin{lem}\thmtitle{Sufficient condition for asymptotic
        stability of a fixed point of the \emph{angle map}}\label{lem:asy_stable_fixed_point} Consider the planar
      system~\eqref{eq:system} under the event-triggering
      rule~\eqref{eq:general_tr_rule}. Assume that the \emph{angle map}
      $\phi(.)$ is continuous. Let $\theta^{\star} \in (0,\pi)$ be a
      fixed point of the \emph{angle map}. If there exists an interval
      $[\bar{\theta}_1,\bar{\theta}_2]$ such that the following
      conditions hold:
     	\begin{itemize}
     	\setlength{\itemsep}{3mm}
     	\item $\theta^{\star} \in (\bar{\theta}_1,\bar{\theta}_2)$
     	\item $\phi(\theta)>\theta$ for all
     	$\theta \in \new{\mathcal{M}}_1 : = [\bar{\theta}_1,\theta^{\star})$ and
     	$\phi(\theta)<\theta$ for all
     	$\theta \in \new{\mathcal{M}}_2 := (\theta^{\star},\bar{\theta}_2]$
     	\item $\phi^2(\theta)>\theta$ for all
     	$\theta \in \new{\mathcal{M}}_1 = [\bar{\theta}_1,\theta^{\star})$
     	\item $[\bar{\theta}_1,\bar{\theta}_2]$ is positively invariant
     	under the $\phi(.)$ map
     \end{itemize}
     then the fixed point $\theta^{\star}$ is asymptotically stable
     and $[\bar{\theta}_1,\bar{\theta}_2]$ is a subset of the region
     of convergence of $\theta^{\star}$.
 \end{lem}

\begin{proof}     
We structure the proof around the following claims.

\textbf{Claim (a):} $\phi^2(\theta) < \theta$ for all
$\theta \in \new{\mathcal{M}}_2$.

\textbf{Claim (b):} Consider an arbitrary
$\theta_0 \in [\bar{\theta}_1,\bar{\theta}_2]$, let
$\theta_k := \phi^k(\theta_0)$. The subsequences
$\{ \theta_k \ | \ \theta_k \in \new{\mathcal{M}}_1 \}$ and
$\{ \theta_k \ | \ \theta_k \in \new{\mathcal{M}}_2 \}$ are strictly increasing
and decreasing, respectively.

We prove Claim~(a) first. Let $\theta \in \new{\mathcal{M}}_2$. By assumption
$\phi(\theta) < \theta$. If $\phi(\theta) \in \new{\mathcal{M}}_2$ then again we
have $\phi^2(\theta) < \phi(\theta) < \theta$, in which case the
claim is true. If $\phi(\theta) = \theta^*$ then again the claim
is true as $\phi^2(\theta) = \theta^* < \theta$. So, the only
remaining case is $\phi(\theta) \in \new{\mathcal{M}}_1$. In this case, we prove
the claim by contradiction. So, suppose
$\phi^2(\theta) \geq \theta > \theta^*$. As $\phi(.)$ is
continuous and $\phi(\theta^*) = \theta^*$, there must exist a
$\bar{\theta} \in [\phi(\theta) , \theta^{\star})$ such that
$\phi(\bar{\theta}) = \theta$ and hence
$\phi^2( \bar{\theta} ) = \phi(\theta)$. But as
$\bar{\theta} \in \new{\mathcal{M}}_1$, we have that
$\phi^2(\bar{\theta}) > \bar{\theta}$. Putting all these
together, we have
$\phi(\theta) = \phi^2(\bar{\theta}) > \bar{\theta} \geq
\phi(\theta)$. This contradiction proves Claim~(a).

Now, we prove Claim~(b) using induction. Given Claim~(a), we
have symmetry in the properties of $\phi(.)$ around
$\theta^*$. Thus, without loss of generality, suppose that
$\theta_0, \ldots, \theta_l \in \new{\mathcal{M}}_1$,
$\theta_{l+1}, \ldots, \theta_{m} \in \new{\mathcal{M}}_2$ and
$\theta_{m+1} \in \new{\mathcal{M}}_1$ for some $l, m \in \natz$ with $m >
l$. Then, we have by assumption that
$\theta_0 < \theta_1 < \ldots < \theta_l < \theta^*$ and
$\theta^* < \theta_m < \ldots < \theta_{l+1}$. Notice that
$\phi(\theta_l) = \theta_{l+1}$, $\phi(\theta^*) = \theta^*$ and
$\phi(.)$ is continuous. So, there must exist a
$\theta \in (\theta_l, \theta^*) \subset \new{\mathcal{M}}_1$ such that
$\phi(\theta) = \theta_m$ and as a result
$\theta_{m+1} = \phi^2(\theta) > \theta > \theta_l$. In this
way, by using induction, and by invoking the symmetry in the
properties of $\phi(.)$ around $\theta^*$, we can conclude that
Claim~(b) is true.

Now, the subsequences $\{ \theta_k \ | \ \theta_k \in \new{\mathcal{M}}_1 \}$
and $\{ \theta_k \ | \ \theta_k \in \new{\mathcal{M}}_2 \}$ may have finite or
infinite length. Both subsequences having finite length can
happen only if the original sequence
$\{\theta_k\}_{k \in \natz}$ hits $\theta^*$ exactly in finite
$k$. Now, notice that these subsequences are bounded and
monotonic and hence must converge to something if they have
infinite length. If one of these subsequences is of finite
length then the limit of the sequence exists and it is
$\theta^*$ as
$\displaystyle \lim_{k \rightarrow \infty} (\phi(\theta_k) -
\theta_k) = 0$ and $\theta^*$ is the only fixed point in
$[\bar{\theta}_1, \bar{\theta}_2]$. If both the subsequences are
infinite then suppose $\{ \theta_k \ | \ \theta_k \in \new{\mathcal{M}}_1 \}$
and $\{ \theta_k \ | \ \theta_k \in \new{\mathcal{M}}_2 \}$ converge to
$a_1 \in \cl{\new{\mathcal{M}}_1}$ and $a_2 \in \cl{\new{\mathcal{M}}_2}$, respectively. But
this can happen only if $\phi^2(a_1) = \phi(a_2) = a_1$ and
which in turn is possible only if $a_1 = a_2 = \theta^*$ since
we have $\phi^2(\theta) \neq \theta$ for all
$\theta \in [\bar{\theta}_1, \bar{\theta}_2] \setminus \{
\theta^*\}$. Finally, notice that if the original sequence
$\{\theta_k\}_{k \in \natz}$ does not hit $\theta^*$ in finite
$k$ then
\begin{equation*}
\{ (k, \theta_k)\}_{k \in \natz} = \{(k, \theta_k) \ | \
\theta_k \in \new{\mathcal{M}}_1 \} \cup \{(k, \theta_k) \ | \ \theta_k \in
\new{\mathcal{M}}_2 \}.
\end{equation*}
Thus, the original sequence $\{\theta_k\}_{k \in \natz}$ also
converges to $\theta^*$.
\end{proof}
  
 	Now, we
    present the main result of this subsection. Note that this result
    is applicable to any scale-invariant event-triggering rule.

  \begin{thm}\label{thm:general_angle_map_behavior}
    \thmtitle{Sufficient condition for a fixed point of the \emph{angle map}
      to be stable or unstable} Consider the planar
    system~\eqref{eq:system} under the event-triggering
    rule~\eqref{eq:general_tr_rule}. Assume that the \emph{angle map}
    $\phi(.)$ is continuous. Let $\theta^{\star} \in (0,\pi)$ be an
    isolated fixed point of the \emph{angle map}. Then, $\theta^{\star}$ is
    stable (asymptotically stable) if the following two conditions are
    satisfied.
    \begin{itemize}
    	\setlength{\itemsep}{3mm}
    \item there exists a neighborhood of $\theta^{\star}$ in which
      $\phi(\theta)-\theta$ decreases (strictly decreases).
    \item there exists $\bar{\theta}<\theta^{\star}$ such that
      $\phi^2(\theta) \ge (>) \theta$ for all
      $\theta \in [\bar{\theta},\theta^{\star})$.
    \end{itemize}
    If there does not exist a neighborhood of $\theta^{\star}$ in
    which $\phi(\theta)-\theta$ decreases, then $\theta^{\star}$ is an
    unstable fixed point of the \emph{angle map}.
  \end{thm}
  	
  \begin{proof}
    We first prove the claim on stability of the fixed point. For each
    $\epsilon > 0$, we can choose $\delta > 0$ small enough (compared
    to $\theta^* - \bar{\theta}$) such that $\phi(\theta)-\theta$
    decreases in $B_\delta(\theta^*)$, $\phi^2(\theta) \ge \theta$ for
    all $\theta \in [\theta^{\star} - \delta, \theta^*)$ and
    \begin{equation*}
      M_\epsilon := \left[ \min_{\theta \in B_\delta(\theta^*) }\{
        \phi(\theta) \}, \max_{\theta \in B_\delta(\theta^*) }\{
        \phi(\theta) \} \right] \in B_\epsilon(\theta^*) .
    \end{equation*}
    The last condition is possible because $\phi(.)$ is continuous and
    $\phi(\theta^*) = \theta^*$. Now, we can show that $M_\epsilon$ is
    positively invariant by similar arguments as in
    Lemma~\ref{lem:asy_stable_fixed_point}. Thus, $\theta^{\star}$ is
    a stable fixed point.

    For the claim on asymptotic stability, notice that for each
    $\epsilon > 0$, we can again construct a neighborhood around
    $\theta^*$ such that the conditions of
    Lemma~\ref{lem:asy_stable_fixed_point} are satisfied. Thus,
    $\theta^{\star}$ is an asymptotically stable fixed point.
  		
    Now, suppose there does not exist a neighborhood of
    $\theta^{\star}$ in which $\phi(\theta)-\theta$ decreases. Note
    that, according to Remark~\ref{rem:finite_fps}, the \emph{angle map} has
    a finite number of fixed points.  Then, atleast one of the
    following two conditions is true. 1) There exists
    $\bar{ \theta}_1 < \theta^{\star}$ such that
    $\phi(\theta)-\theta < 0$ for all
    $\theta \in [\bar{ \theta}_1, \theta^{\star})$ or 2) there exists
    $\bar{ \theta}_2 > \theta^{\star}$ such that
    $\phi(\theta)-\theta > 0$ for all
    $\theta \in (\theta^{\star},\bar{ \theta}_2]$. In both the cases,
    we can show the existence of an $\epsilon >0$ such that for any
    $\delta \in (0,\epsilon]$, there exists $\theta_0$ in the
    $\delta-$neighborhood of $\theta^{\star}$ such that the sequence
    $\{\theta_k\}_{k \in \natz}$ generated by the $\phi(.)$ map exits
    the $\epsilon-$neighborhood of $\theta^{\star}$ for some
    $k \in \nat$. This implies that $\theta^{\star}$ is an unstable
    fixed point.
    \end{proof}
  	
  	Note that Lemma~\ref{lem:asy_stable_fixed_point} and Theorem~\ref{thm:general_angle_map_behavior} only require
  	the function $\phi(.)$ to be continuous and not differentiable,
  	unlike the existing results in the literature on the stability of fixed points of a nonlinear map. As we do not require Assumption~\ref{A:M} in these results, the angle map is not necessarily differentiable, even if it is continuous. Note that, even if Assumption~\ref{A:M} holds, the ``min" 
  	operator in the definition~\eqref{eq:IET_theta} may introduce a point where the 
  	inter-event time function, and hence the angle map, is continuous 
  	but not differentiable.
  	
      \begin{cor}\thmtitle{\emph{angle map} with pairs of stable and unstable
        fixed points} \label{rem:general_angle_map_behavior}
      Consider the planar system~\eqref{eq:system} under the
      event-triggering rule~\eqref{eq:general_tr_rule}. Assume that
      the \emph{angle map} $\phi(.)$ is continuous \new{and $\phi(.)$ has} a
      set of \new{even number of} fixed points $\{\theta_1,\theta_2\ldots,\theta_{2\new{l}}\}$, \new{for some $l \in \nat$}, in
      the interval $[0,\pi)$ where
      $\theta_i < \theta_{i+1} \quad \forall i \in
      \{1,2,\ldots,2\new{l}-1\}$. Let $\phi(\theta)>\theta$ for all
      $\theta \in (\theta_{2i-1},\theta_{2i})$ and
      $\phi(\theta)<\theta$ for all
      $\theta \in (\theta_{2i},\theta_{2i+1})$ for all
      $i \in \{1,2,\ldots,\new{l}\}$ where
      $\theta_{2\new{l}+1}=\theta_1+\pi$. Then $\theta_{2i-1}$ is an
      unstable fixed point of the \emph{angle map}
      $\forall i \in \{1,2,..,\new{l}\}$. Assume also that the interval
      $[\theta_{2i-1},\theta_{2i+1}]$ is invariant under the \emph{angle map}
      $\forall i \in \{1,2,..\new{l}\}$. If $\phi^2(\theta)>\theta$ for all
      $\theta \in (\theta_{2i-1},\theta_{2i})$, then $\theta_{2i}$ is
      an asymptotically stable fixed point and the region of
      convergence is $(\theta_{2i-1},\theta_{2i+1})$ for all
      $ i \in \{1,2,\ldots,\new{l}\}$. \remend
    \end{cor}
  	
    \begin{proof}
    	Note that, for all $i \in \{1,2,\ldots \new{l}\}$, $\theta_{2i}$ satisfies the conditions of Lemma~\ref{lem:asy_stable_fixed_point} and $\theta_{2i-1}$ satisfies the instability conditions of Theorem~\ref{thm:general_angle_map_behavior}.
      Thus, proof of this result follows directly from
      Lemma~\ref{lem:asy_stable_fixed_point} and
      Theorem~\ref{thm:general_angle_map_behavior}.
    \end{proof}

 \new{In numerical examples, we have often observed that the \emph{angle map}
has even number of fixed points in the interval
$[0,\pi)$. In corollory~\ref*{rem:general_angle_map_behavior}, we provide some analytical guarantees for this behavior.}
   
  Remark~\ref{rem:IET-analysis} and Theorem~\ref{thm:IET-analysis} help to
  analyze the evolution of inter-event times under the general class
  of event-triggering rules~\eqref{eq:general_tr_rule}. But, it is difficult to say
  anything more specific that holds for all the triggering
  rules. Thus, in the following subsection, we consider the specific
  event-triggering rule~\eqref{eq:tr2}, or
  equivalently~\eqref{eq:general_tr_rule} with $M(.) = M_2(.)$ given
  in~\eqref{eq:M2}. We analyze the inter-event times that are
  generated by this rule for the planar system~\eqref{eq:system}.
  
  \subsection{Analysis of fixed points of $\phi(.)$ with %
    $M(.)=M_2(.)$}

  Here, our goal is to provide necessary and sufficient conditions for the existence
  of a fixed point for the \emph{angle map} $\phi(.)$ under the specific
  event-triggering rule~\eqref{eq:tr2} or
  equivalently~\eqref{eq:general_tr_rule} with $M(.) = M_2(.)$ given
  in~\eqref{eq:M2}. First, in the following lemma, we present a
  necessary and sufficient condition on a function of time that must be satisfied if
  the \emph{angle map} is to have a fixed point. Building on this lemma,
  we then present an algebraic necessary condition.
  
  \begin{lem}\thmtitle{Necessary and sufficient condition for the \emph{angle map} to have a
      fixed point under triggering
    rule~\eqref{eq:tr2}} \label{lem:fixed_point}
  Consider the planar system~\eqref{eq:system} under
  the event-triggering rule~\eqref{eq:tr2} or
  equivalently~\eqref{eq:general_tr_rule} with $M(.) = M_2(.)$ given
  in~\eqref{eq:M2}. Suppose that 
  the parameter $\sigma \in (0, 1)$ is such that the origin of the
  closed loop system is globally asymptotically stable. Then, there
  exists a fixed point for the \emph{angle map} $\phi(.)$ if and only
  if $\det(L(\tau))=0$ for some $\tau \in \realp$ and there
    exists $x \neq 0$ in the nullspace of $L(\tau)$ such that
    $\tau_e(x)=\tau$, where
$L(\tau)$ is defined as in~\eqref{eq:L_tau} with $\alpha=(1+\sigma)^{-1}$.
  \end{lem}

  \begin{proof}
    There exists a fixed point for the $\phi(\theta)$ map if and
      only if there exists an $x \in \nzrealtwo$ such that $x(t_k)=x$
    implies $x(t_{k+1})=\alpha x$ for some $\alpha > 0$. Note
    that $\alpha$ cannot be negative because then
    $\norm{ x(t_k) - x(t_{k+1}) } = (1 - \alpha^{-1})
    \norm{x(t_{k+1})} > \norm{x(t_{k+1})}$, which is not possible for
    the event-triggering rule~\eqref{eq:tr2} with $\sigma \in (0,
    1)$. Further, if $\alpha > 1$ then for the initial condition
    $x(t_0) = x$, $x(t_{k}) = \alpha^k x$ would grow unbounded, which
    violates the assumption that $\sigma$ is such that the
    event-triggering rule~\eqref{eq:tr2} guarantees global asymptotic
    stability. Thus, it must be that $\alpha \in (0,1)$. Using this
    information, from the event-triggering rule~\eqref{eq:tr2}, we
    obtain $\alpha=(1+\sigma)^{-1}$. Now, we can express
    \begin{equation*}
      x(t_{k+1}) = G(\tau') x(t_k) = \alpha x(t_k),
    \end{equation*}
    where $\tau'=\tau_e(x(t_k))$. This is possible if and only if
      $\tau'=\tau_e(x(t_k))$ and $L(\tau')x(t_k) = 0$. In this
      case, $\det(L(\tau'))=0$.
  \end{proof}

\new{Note that Lemma~\ref{lem:fixed_point} is an extension of Lemma 11 
in our conference paper~\cite{AR-PT:2020}, where we only provide a 
necessary condition for the angle map to have a fixed point under the 
triggering rule~\eqref{eq:tr2}. Lemma~\ref{lem:fixed_point} is similar 
to Proposition 6 in~\cite{GG-MM:2022}, but not the same. Notice from 
the proof of Lemma~\ref{lem:fixed_point} that $\det(L(\tau'))=0$ 
(equivalently that $G(\tau')$ has eigenvalue $\alpha$) is not 
sufficient for the angle map to have a fixed point. This is because for 
an $x \neq 0$ in the nullspace of $L(\tau')$, $f(x, \tau) = 0$ may have 
multiple solutions $\tau$ and hence $\tau_e(x)$ may be strictly less 
than $\tau'$. This subtlety is not addressed in Proposition 6 
in~\cite{GG-MM:2022} or its proof.}

  While Lemma~\ref{lem:fixed_point} provides a necessary and sufficient condition for
  the \emph{angle map} $\phi(.)$ to have a fixed point, it may not be easy to
  verify if $\det(L(\tau))=0$ for some $\tau \in [\tmin,
  \tmax]$. Thus, we next present an algebraic necessary condition for
  the existence of fixed points for the \emph{angle map} $\phi(.)$.

  \begin{prop}\thmtitle{Algebraic necessary condition for the \emph{angle map}
      to have a fixed point under triggering
    rule~\eqref{eq:tr2}} \label{prop:nec-cond-fixed-point}
  Consider the planar system~\eqref{eq:system} under
  the event-triggering rule~\eqref{eq:tr2} or
  equivalently~\eqref{eq:general_tr_rule} with $M(.) = M_2(.)$ given
  in~\eqref{eq:M2}. Suppose that the parameter $\sigma \in (0, 1)$ is such that the origin of
  the closed loop system is globally asymptotically stable. Further,
  assume that both the eigenvalues of $A$ have positive real
  parts. Let $A \rdef S J S^{-1}$, where $J \in \real^{2 \times 2}$ is
  the real Jordan form of $A$. Let
  \begin{equation*}
    R \ldef S^{-1} \left[ I - (1-\alpha)A A_c^{-1} \right] S \quad \text{with} \quad \alpha=(1+\sigma)^{-1},
  \end{equation*}
  \begin{equation*}
    \sigma_m(\tau) \ldef \e^{\lambda \tau} \sqrt{\frac{(\tau^2+2)-\tau
        \sqrt{\tau^2+4}}{2}} .
  \end{equation*}
  Then, there exists a fixed point for the \emph{angle map} $\phi(.)$ only if
  \begin{itemize}
  	\setlength{\itemsep}{3mm}
  \item $\norm{R} > 1$, if either $A$ is non-diagonalizable with
    eigenvalue $\lambda \ge 0.5$ or $A$ is diagonalizable.
  \item
    $\norm{R} \ge \sigma_m\left( \sqrt{\frac{1}{\lambda^2}-4} \right)
    $, if $A$ is non-diagonalizable with eigenvalue
    $\lambda \in (0, 0.5)$.
  \end{itemize}
  \end{prop}

  \begin{proof}
    First note that
    \begin{equation*}
      AL(\tau) = (1 - \alpha)A + (\e^{A \tau} - I) A_c .
    \end{equation*}
    Suppose there exists a fixed point for the $\phi(\theta)$ map.
    Then by Lemma~\ref{lem:fixed_point}, we know that there exists a
    $\tau \in \realp$ such that $L(\tau) x_0 = 0$ for some
    $x_0 \in \nzrealtwo$. This implies that $AL(\tau) x_0 = 0$
      for some $x_0 \in \nzrealtwo$ and $\tau \in \realp$.  However
    this is equivalent to saying
    \begin{equation*}
      \left( \e^{A \tau} - I \right) z_0 = - (1 - \alpha) A A_c^{-1}
      z_0, \quad z_0 = A_c x_0 .
    \end{equation*}
    Note that $A_c$ is invertible because we have assumed it is
    Hurwitz. Thus, there exists a vector $v \ldef S^{-1} A_c x_0$ such
    that
    \begin{equation}\label{eq:L_R}
      \e^{J \tau} v = R v, \ \text{for some }
      \tau > 0 .
    \end{equation}
    Note that
    $\norm{ \e^{J \tau} v } \ge \sigma_{min}(\e^{J \tau})\norm{v}$,
    where $\sigma_{min}(\e^{J \tau})$ denotes the minimum singular
    value of $\e^{J \tau}$. Thus,
    $\norm{R} \ge \sigma_{min}(\e^{J \tau})$ for some $\tau>0$. Recall
    that we assumed that both the eigenvalues of $A$ have positive
    real parts. We can show that if $A$ is diagonalizable and has real
    positive eigenvalues, then
    $\sigma_{min}(\e^{J \tau})=\e^{\lambda_1 \tau}$ where $\lambda_1$
    is the minimum eigenvalue of $A$. Similarly, if $A$ has complex
    conjugate eigenvalues then
    $\sigma_{min}(\e^{J \tau}) = \e^{\mu \tau}$ where $\mu > 0$ is the
    real part of the eigenvalues. On the other hand, if $A$ is
    non-diagonalizable with eigenvalue $\lambda$, then
    $\sigma_{min}(\e^{J \tau})=\e^{\lambda \tau}
    \sqrt{\frac{(\tau^2+2)-\tau \sqrt{\tau^2+4}}{2}}\rdef
    \sigma_m(\tau)$. If $A$ is non-diagonalizable with eigenvalue
    $\lambda \ge 0.5$ or if $A$ is diagonalizable,
    $\sigma_{min}(\e^{J \tau})$ is a monotonically increasing function
    of $\tau$. Thus, $\sigma_{min}(\e^{J \tau})>1$ for all
    $\tau>0$. If $A$ is non-diagonalizable with eigenvalue
    $\lambda \in (0, 0.5)$, $\sigma_{min}(\e^{J \tau})$ attains a
    minimum value when $\tau=\sqrt{\frac{1}{\lambda^2}-4}$. Thus,
    $\sigma_{min}(\e^{J \tau}) \ge \sigma_m \left(
      \sqrt{\frac{1}{\lambda^2}-4} \ \right)$ for all $\tau>0$. This
    completes the proof of the result.
  \end{proof}

  Next we show that the algebraic necessary condition for the \emph{angle map} to have a fixed point under event-triggering rule~\eqref{eq:tr2}
  is always satisfied if $A$ is diagonalizable with eigenvalues having
  positive real parts and $A_c$ has real negative eigenvalues.

\begin{prop} \label{lem:nec-cond-fixed-point-Ac-with-real-evs}%
  Consider planar system~\eqref{eq:system} under the event-triggering
  rule~\eqref{eq:tr2} or equivalently~\eqref{eq:general_tr_rule} with
  $M(.) = M_2(.)$ given in~\eqref{eq:Ms}. Suppose that the parameter
  $\sigma \in (0, 1)$ is such that the origin of the closed loop
  system is globally asymptotically stable. Further, assume that both
  the eigenvalues of $A$ have positive real parts and $A$ is
  diagonalizable. Let $A \rdef S J S^{-1}$, where $J$ is the real
  Jordan form of $A$ and let $A_c$ have real negative
  eigenvalues. Then, $\norm{R} > 1$, where
  \begin{equation*}
    R \ldef S^{-1} \left[ I - (1-\alpha)A A_c^{-1} \right] S . 
  \end{equation*}
\end{prop}
\begin{proof}
  Note that the induced 2-norm of matrix $R$ can be expressed as
  $\norm{R}=\sup\{u^TRv:\norm{u}=\norm{v}=1,u,v \in \real^2\}$. Let
  $(\lambda_c,x)$ be an eigen-pair of $A_c$, where $\lambda_c<0$. Let
  $y$ be a unit vector defined as
  $y \ldef \frac{S^{-1}x}{\norm{S^{-1}x}}$. Then,
  \begin{align*}
    \norm{R} &\ge y^TRy = y^Ty-(1-\alpha)y^TJ S^{-1}A_c^{-1}Sy \\
             &=1-\frac{1-\alpha}{\lambda_c}y^TJ y 
               \ge 1-\frac{1-\alpha}{\lambda_c} \lambda >1
  \end{align*}
  where $\lambda=\lambda_{\min}(A) > 0$ if $A$ has real eigenvalues or
  $\lambda=\text{Re}(\lambda(A)) > 0$ if $A$ has complex conjugate
  eigenvalues. For obtaining the last inequality, we have used the
  facts that $\alpha \in (0, 1)$ (see proof of
  Lemma~\ref{lem:fixed_point}) and that $\lambda_c < 0$.
\end{proof}

Next, we present a geometric interpretation of the event-triggering
rule~\eqref{eq:tr2}, which we then use to give bounds on the
difference between an angle $\theta$ and $\phi(\theta)$.

\begin{rem}\label{rem:circle_interpretation} \thmtitle{Geometric
		interpretation of the event-triggering rule~\eqref{eq:tr2}} The
	locus of points $x$ which satisfy the equation
	$\norm{x-\hat{x}}=\sigma \norm{x}$ for a fixed $\hat{x}$ and
	$\sigma$ is a circle with center at $\frac{\hat{x}}{1-\sigma^2}$ and
	radius $\frac{\sigma}{1-\sigma^2}\norm{\hat{x}}$. Also note that
	origin is always outside this circle. Hence, the event-triggering
	rule~\eqref{eq:tr2} ensures that for all $k \in \natz$, $x(t_k)$ and
	$x(t_{k+1})$ satisfy
	$\norm{x(t_{k+1})-\frac{x(t_k)}{1-\sigma^2}}=\frac{\sigma}{1-\sigma^2}
	\norm{x(t_k)}$.  \remend
\end{rem}

This observation leads us to an upper bound for
$|\phi(\theta_{k})-\theta_{k}|, \ \forall k \in \natz$, which we
present in the following result. Note that this bound is useful
  from a computational point of view for determining the fixed points
  of the \emph{angle map}.

\begin{lem}\thmtitle{Upper bound on
    $|\phi(\theta_{k})-\theta_{k}|$} \label{lem:phi-theta_bound}
  Consider the planar system~\eqref{eq:system} under the
  event-triggering rule~\eqref{eq:tr2} or
  equivalently~\eqref{eq:general_tr_rule} with $M(.) = M_2(.)$ given
  in~\eqref{eq:Ms}. Suppose that the parameter $\sigma \in (0, 1)$ is
  such that the origin of the closed loop system is globally
  asymptotically stable. Then the evolution of the ``angle'' of the state
  from one sampling time to the next is uniformly bounded by
  $\sin^{-1}\left(\sigma\right)$. That is,
  $|\phi(\theta_{k})-\theta_{k}|\le \sin^{-1}\left(\sigma\right),
  \quad \forall k \in \natz$.
\end{lem}

\begin{proof}
  According to the geometric interpretation of the event-triggering
  rule~\eqref{eq:tr2} provided in
  Remark~\ref{rem:circle_interpretation}, $x(t_{k+1})$ is on the
  circle with center at $\frac{x(t_k)}{1-\sigma^2}$ and radius
  $\frac{\sigma}{1-\sigma^2}\norm{x(t_k)}$. Thus, the angle between
  $x(t_k)$ and $x(t_{k+1})$ is the maximum when $x(t_{k+1})$ is on a
  tangent to this circle that passes through the origin. As a tangent
  to the circle is perpendicular to the radial line passing through
  the point of tangency, the maximum possible angle is exactly equal
  to $\sin^{-1}\left(\sigma\right)$. That is,
  $|\phi(\theta_{k})-\theta_{k}|\le \sin^{-1}\left(\sigma\right), \
  \forall k \in \natz$. 
\end{proof}

\begin{rem}\label{rem:RT_non-convergence}
	In the event-triggered control literature, typically, the relative thresholding parameter $\sigma \in (0,1)$ in the event-triggering rule~\eqref{eq:tr2} is
	such that the origin of the closed loop system is globally
	asymptotically stable. Then,  $|\phi(\theta_{k})-\theta_{k}| \leq 
	\sin^{-1}\left(\sigma\right) < \frac{\pi}{2}, \
	\forall k \in \natz$. According to Theorem~\ref{thm:IET-analysis}, this implies that, the inter-event times converge to a steady state value if and only if the angle map has a fixed point.
\end{rem}

\section{Asymptotic average inter-event time}\label{sec:Asymptotic-avg-inter-event-time}

In this section, we analyze the asymptotic average inter-event time as
a function of the angle of the initial state of the
system~\eqref{eq:system} under the event-triggering
rule~\eqref{eq:general_tr_rule}. First, we study ergodicity of the
\emph{angle map} and then, with the help of ergodic theory, we provide a
sufficient condition for the asymptotic average inter-event time to be
a constant for all non-zero initial conditions of the system
state. Later, with the help of rotation theory, we analyze the
asymptotic behavior of the inter-event times, such as convergence or
non-convergence to a periodic orbit, for a special case where the
\emph{angle map} is an orientation preserving homeomorphism. 
Note that, in this section, we do not provide any fundamentally 
new results. Rather, we invoke the existing results in ergodic theory 
and rotation theory to provide a mathematical explanation for different 
kinds of asymptotic behavior of the inter-event times. We make the
following assumption in this section of the paper.

\begin{enumerate}[resume,label=\textbf{(A\arabic*)},align=left]
\item \hspace{.1em} The inter-event time function $\tau_s(.)$, defined
  as in~\eqref{eq:IET_theta}, is continuous. \label{A:continuous
    tau_s}
\end{enumerate}

Assumption ~\ref{A:continuous tau_s} is not very restrictive as in
Theorem~\ref{thm:ts-continuity} we show, under mild technical
assumptions on $M(\tau)$, that in general the inter-event time
function $\tau_s(.)$ is continuous except at finitely many angles
$\theta$. In Corollary~\ref*{cor:suff-cond-ts-continuous}, we also
provide a sufficient condition under which the function $\tau_s(.)$ is
continuous. Note also that the \emph{angle map} $\phi(.)$ is a continuous
map on a compact metric space under Assumption~\ref{A:continuous
  tau_s}.

\subsection{Ergodicity of the \emph{angle map}}\label{sec:angle-map-ergodicity}

In this subsection, we study about the ergodicity of the \emph{angle map} to analyze the asymptotic average inter-event time as a function
of the initial state of the system.

\begin{rem}\thmtitle{\emph{angle map} is ergodic under Assumption~\ref{A:continuous tau_s}}\label{rem:T_ergodic}
  Consider system~\eqref{eq:system} under the event-triggering
  rule~\eqref{eq:general_tr_rule}. Let Assumption~\ref{A:continuous
    tau_s} hold. Then according to Krylov-Bogolyubov
  theorem~\cite{NK-NB:1937}, there exists an invariant probability
  measure under the \emph{angle map} $\phi:R^1 \to R^1$, defined as
  in~\eqref{eq:Tmap}, as it is a continuous map on the compact space
  $R^1$. Moreover, according to Theorem 4.1.11 in~\cite{AK-BH:1995},
  there exists at least one ergodic measure in the set of all
  $\phi-$invariant probability measures. Hence, the \emph{angle map} is
  ergodic under Assumption~\ref{A:continuous tau_s}.\remend
\end{rem}

Now, we define the asymptotic average inter-event time function, $\tavg$, as
\begin{equation}\label{eq:tau_avg}
\tavg(\theta) \ldef \lim_{k \to \infty}\frac{1}{k}\sum_{j=0}^{k-1}\tau_s(\phi^j(\theta)).
\end{equation}
Note that in general, this function may not be defined for every
$\theta \in R^1$, but for the $\theta$ for which the limit exists,
$\tavg(\theta)$ denotes the asymptotic average inter-event time when
the angle of the initial state of the system is $\theta$. Now,
based on the Birkhoff Ergodic theorems~\cite{PW:1982}, we can say the
following regarding $\tavg(\theta)$.

\begin{lem}\thmtitle{Asymptotic average inter-event time function is a
    constant almost
    everywhere} \label{lem:asy_avg_iet_constant_almost_everywhere}
  Consider system~\eqref{eq:system} under the event-triggering
  rule~\eqref{eq:general_tr_rule}. Let Assumption~\ref{A:continuous
    tau_s} hold. Let the \emph{angle map} $\phi:R^1 \to R^1$, defined as in~\eqref{eq:Tmap}, be ergodic
  on the probability space $(R^1,\mathcal{B},\mu)$. Then the
  asymptotic average inter-event time function $\tavg$ defined as in
  ~\eqref{eq:tau_avg} exists for $\mu$-almost every $\theta \in
  R^1$. Moreover, $\tavg(\theta)=\int \tau_s d\mu$ for $\mu$-almost
  every $\theta \in R^1$.
\end{lem}
\begin{proof}
  Proof of this lemma follows directly from the Birkhoff ergodic
  theorems for measure preserving transformations and ergodic
  transformations, respectively.
\end{proof}

Ergodicity of the \emph{angle map} implies that the asymptotic average
inter-event time function is a constant almost everywhere, with
respect to the measure $\mu$, on $R^1$. However, $\tavg$ may be
different in each invariant set of the \emph{angle map}. Now, we provide
a sufficient condition for the uniform convergence of average
inter-event time to a constant for every point on $R^1$.

\begin{prop}\thmtitle{Sufficient condition for the uniform convergence of average inter-event time to a constant for every point on $R^1$} \label{thm:unique_ergodicity}
  Consider system~\eqref{eq:system} under the event-triggering
  rule~\eqref{eq:general_tr_rule}. Let Assumption~\ref{A:continuous
    tau_s} hold. If the \emph{angle map}, $\phi:R^1 \to R^1$ defined as
  in ~\eqref{eq:Tmap}, has at most one periodic orbit, then
  the average inter-event time converges uniformly to a constant for
  every point in $R^1$.
\end{prop}
\begin{proof}
  According to the theorem in \cite{CA-1992}, the \emph{angle map},
  $\phi:R^1 \to R^1$, is uniquely ergodic if and only if $\phi$ has at
  most one periodic orbit. Hence, by the Oxtoby ergodic
  theorem~\cite{PW:1982}, if the \emph{angle map} has at most one periodic orbit then the average inter-event time converges
  uniformly to a constant for every point on $R^1$.
\end{proof}
Note that, Proposition~\ref{thm:unique_ergodicity} provides a sufficient
condition for the asymptotic average inter-event time function,
$\tau_{\text{avg}}(\theta)$ to be a constant for all $\theta \in R^1$.
This result also suggests that the analysis of periodic orbits of the
\emph{angle map} helps in analyzing the asymptotic average inter-event
time function.

\subsection{Asymptotic behavior of the inter-event times}\label{sec:average-iet-analysis}
In this subsection, we consider a special case where the \emph{angle map} is an orientation preserving homeomorphism. 
\begin{enumerate}[resume,label=\textbf{(A\arabic*)},align=left]
\item \hspace{.1em}The \emph{angle map}, $\phi: R^1 \to R^1$ defined as
  in~\eqref{eq:Tmap}, is an orientation-preserving
  homeomorphism. \label{A:T homeomorphism}
\end{enumerate}
Note that, the \emph{angle map} is said to be a homeomorphism if it is
continous and bijective with a continuous inverse. The \emph{angle map}
is said to be orientation-preserving if it admits a monotonically
increasing lift. Note that, the \emph{angle map} is a homeomorphism if it is continuous and orientation-preserving. Under this special case, we provide a framework to
analyze the asymptotic behavior of the inter-event times, such as
convergence or non-convergence to a periodic orbit, with the help of
rotation theory. This analysis also gives us insights into the
asymptotic average inter-event time as a function of the initial
state.

Now, let $\bar{\pi}: \real \to R^1$ be defined as
$\bar{\pi}(x) = x(\text{mod}\hspace{1ex} 2\pi)$, i.e., the projection
of the real line onto $R^1$. Let $\Phi: \real \to \real$ be a lift of
$\phi$, that is
$(\bar{\pi} \circ \Phi)(x) = (\phi \circ \bar{\pi})(x)$ for all
$x \in \real$. Next, we define the rotation number of the \emph{angle map},
\begin{equation}\label{eq:angle_map_rotation_number}
\bar{\rho}(\phi) \ldef \bar{\pi}(\rho(\Phi)),
\end{equation}
where $\rho(.)$ is defined as, 

\begin{equation*}
\rho (\Phi)=\lim _{n\to \infty }{\frac {\Phi^{n}(x)-x}{n}}.
\end{equation*}
Note that, according to Proposition 11.1.1 in~\cite{AK-BH:1995}, as
$\phi(.)$ is an orientation-preserving homeomorphism of the circle,
this limit exists for every $x \in \real$ and is independent of the
point $x$. The rotation number plays a crucial role in determining the
qualitative behavior of the orbits of an orientation-preserving
homeomorphism. We can determine the existence of a periodic point of
an orientation-preserving homeomorphism if we know the rationality of
the rotation number of the map. Thus, the rationality of the rotation
number of the \emph{angle map} indirectly helps us to comment about the
uniform convergence of the average inter-event time to a constant for
every point on $R^1$.

\begin{prop}\thmtitle{\emph{angle map} with irrational rotation number}\label{thm:irrational_angle_map}
  Consider system~\eqref{eq:system} under the event-triggering
  rule~\eqref{eq:general_tr_rule}. Let Assumptions~\ref{A:continuous
    tau_s} and~\ref{A:T homeomorphism} hold. If
  the rotation number of the \emph{angle map}, defined as
  in~\eqref{eq:angle_map_rotation_number}, is irrational, then the
  average inter-event time converges to a constant uniformly for all
  initial states of the system. Moreover, the $\omega-$limit set
  $\omega(\theta)$ is independent of $\theta \in R^1$ and is either
  $R^1$ or perfect and nowhere dense.
\end{prop}

\begin{proof}
  According to Proposition 11.1.4 and Proposition 11.1.5
  in~\cite{AK-BH:1995}, if the rotation number of the \emph{angle map} is
  irrational then there does not exist a periodic orbit for the
  \emph{angle map}. Hence, by Proposition~\ref{thm:unique_ergodicity}, the
  average inter-event time converges to a constant uniformly for all
  initial states of the system. Moreover, according to Proposition
  11.2.5 in~\cite{AK-BH:1995}, the $\omega-$limit set $\omega(\theta)$
  is independent of $\theta \in R^1$ and is either $R^1$ or perfect
  and nowhere dense.
\end{proof}

Rotation theory also helps us to describe the qualitative behavior of the orbits of an orientation-preserving homeomorphism with rational rotation number.

\begin{prop}\thmtitle{\emph{angle map} with rational rotation number}\label{thm:rational_angle_map}
  Consider system~\eqref{eq:system} under the event-triggering
  rule~\eqref{eq:general_tr_rule}. Let Assumptions~\ref{A:continuous
    tau_s} and~\ref{A:T homeomorphism} hold. If
  the rotation number of the \emph{angle map}, defined as
  in~\eqref{eq:angle_map_rotation_number}, is rational,
  $\bar{\rho}(\phi)=\frac{p}{q}$, then every forward orbit of $\phi$
  converges to a periodic sequence with period $q$. Moreover, for all
  initial states of the system, inter event times converges to a
  periodic orbit with period $q$.
\end{prop}

\begin{proof}
  Proof of this proposition follows directly from Proposition 11.1.4,
  Proposition 11.1.5 and the \emph{Poincare classification}
  in~\cite{AK-BH:1995}.
\end{proof}

\begin{rem} \thmtitle{Number of periodic points and periodic
    orbits} %
  If the rotation number of $\phi$ is rational,
  $\bar{\rho}(\phi)=\frac{p}{q}$, then $\phi^q$ map has $mq$ fixed
  points where $m \in \nat$ denotes the number of periodic
  orbits of the \emph{angle map}. If there are $m > 1$ periodic orbits,
  without semi-stable periodic orbits, then $m$ is always
  even. \remend
\end{rem}

\begin{rem}\thmtitle{Stability of periodic orbits and $\tavg(\theta)$
    in the region of convergence of a stable periodic orbit}
Note that, if we know the rotation number of the \emph{angle map} precisely
and if it is a rational number, $\bar{\rho}(\phi)=\frac{p}{q}$, then we can determine
the periodic orbits of the \emph{angle map} by analyzing the
fixed points of $\phi^q(.)$ map. Further, we can analyze stability
and region of convergence of periodic orbits of $\phi(.)$ map by
analyzing stability and region of convergence of fixed points
of $\phi^q(.)$ map. Note also that, the asymptotic average inter-event
time $t_{avg}(\theta)$ is a constant for all $\theta$ in the region of
convergence of a stable periodic orbit. And this constant is
equal to the average of $\tau_s(\theta)$, averaged over the finitely many
(q number of them) periodic points $\theta$ in the corresponding
periodic orbit.
  \remend
\end{rem}

Note that, we can easily analyze stability and region of convergence of periodic orbits $\phi(.)$ map by using Lemma~\ref{lem:asy_stable_fixed_point} and Theorem~\ref{thm:general_angle_map_behavior}.

There are a number of papers, see for
example~\cite{VV-1988,RP-1995,AB-2016}, which propose different
methods to determine a numerical approximation of the rotation number
of a circle homeomorphism. In these papers, authors claim that it is
possible to check the rationality of the rotation
number. Specifically, if the rotation number is rational then it is
possible to find the exact value in finite number of steps. But one
drawback is, if the rotation number is irrational then these
algorithms will not terminate in finite number of steps. Moreover, due
to the inevitability of numerical errors in finite-precision
arithmetic on computers, in general it is practically not possible to
conclusively determine if the rotation number is rational or
irrational. Nevertheless, results that we have presented in this paper
are still valuable as they provide a mathematical explanation for different kinds of asymptotic behavior of the orbits of the
angle as well as the inter-event times. For the situations in which we
know that the \emph{angle map} has a fixed point or that it has periodic
points of a certain period then one could use the insights from our
results to determine $\tavg(\theta)$ for different $\theta$ in an
efficient manner.

\section{Numerical examples} \label{sec:numerical-examples}

In this section, we illustrate our results and highlight some
interesting behavior of the inter-event times using numerical examples
of several different systems of the form~\eqref{eq:system} with the
event-triggering rule~\eqref{eq:tr1} and~\eqref{eq:tr2}. In each case, we
choose the control gain matrix $K$ so that $A_c = (A+BK)$ is
Hurwitz. We choose a quadratic Lyapunov function $V(x) = x^T P x$,
where $P$ is the solution of the Lyapunov equation~\eqref{eq:lyap}
with $Q=I$. Then, from an analysis as in~\cite{PT:2007}, we set the
thresholding parameter
$\sigma=\frac{0.99\lambda_{\min}(Q)}{2\norm{PBK}}$ in the
event-triggering rule~\eqref{eq:tr2} for the sampled data
controller~\eqref{eq:control_input}. Next, we describe specific cases
in detail.

\subsubsection*{Case 1:}\hspace{0.1em}
Consider the system
\begin{equation*}
\dot{x}=\begin{bmatrix*}[r]
0 & 1 \\ -2 & 3
\end{bmatrix*}x+\begin{bmatrix} 1 & 0 \\ 0 & 1
\end{bmatrix}u \rdef Ax + Bu .
\end{equation*}
The system matrix $A$ has real eigenvalues at $\{1, 2\}$. We let the
control gain %
$K =
\begin{bmatrix*}[r]
  -1 & -0.8 \\ 3 & -4
\end{bmatrix*}$ %
so that $A_c$ has real eigenvalues at $\{-0.5528,
-1.4472\}$. Figure~\ref{fig:Ac_with_real_evs} shows the simulation
results of this system for the event triggering
	rule \eqref{eq:tr2}. Figure~\ref{fig:case1_LmdM} presents the evolution of the smallest eigenvalue of the time-varying symmetric matrix $\dot{M}$ and it shows that the sufficient condition for continuous differentiability of $\tau_s(.)$, given in Corollary~\ref{cor:suff-cond-ts-continuous}, is satisfied. Hence, for this case, the inter-event time function
$\tau_s(\theta)$ is continuously differentiable and it is also periodic with period $\pi$. From
Figure~\ref{fig:case1_detM} and Figure~\ref{fig:case1_tau} we can
verify that $\det(M(\tau))=0$ has exactly two solutions and they are
precisely $\tmin$ and $\tmax$
respectively. Figure~\ref{fig:case1_detL} shows that there are two
points at which $\det(L(\tau))=0$. Figure~\ref{fig:case1_angle_map}
verifies that the \emph{angle map} $\phi(.)$ has exactly two fixed points at
$\theta_1=1.15$rad and $\theta_2=1.85$rad in the interval $[0,\pi)$,
where $\theta_1$ is a stable fixed point. Figure~\ref{fig:Lift_1}
presents a lift of the \emph{angle map}. As the lift $\Phi(.)$ is
increasing monotonically, the \emph{angle map} $\phi(.)$ is an
orientation preserving homeomorphism. Based on our analysis, we can
say that there does not exist any periodic orbit with period greater
than one. We also know that every forward orbit of the \emph{angle map}
converges to one of the fixed
points. Figure~\ref{fig:case1_phase_portrait} is the phase portrait of
the closed loop system. Notice that the state trajectories converge to
a radial line which makes an angle of $1.15$ radian with the positive
$x_1$ axis, which is exactly the point at which the \emph{angle map}
$\phi(.)$ has the stable fixed point. From
Figure~\ref{fig:case1_tau_k} it is clear that, for multiple values for
the initial state of the system, the inter-event time converges to a
steady state value of $0.18$, which is exactly equal to
$\tau_s(\theta_1)$. Figure~\ref{fig:tau_avg1} presents the average
inter-event time, evaluated for different values of total number of
sampling instants, as a function of the angle of the initial state
of the system. Note that as the total number of sampling instants
($N$) increases, the average inter-event time, for all initial
conditions except the case where the angle of the initial state of
the system is an unstable fixed point of the \emph{angle map}, converges
to the value of inter-event time function at the stable fixed points
of the \emph{angle map}. Due to the error in numerical computations, as
$N$ increases, the \emph{computed} value of the average inter-event
time at the unstable fixed points of the \emph{angle map} diverges from
the actual value of the inter-event time function at those points.

\begin{figure}[ht]
  \centering
  \begin{subfigure}{4cm}
  	\includegraphics[width=\textwidth]{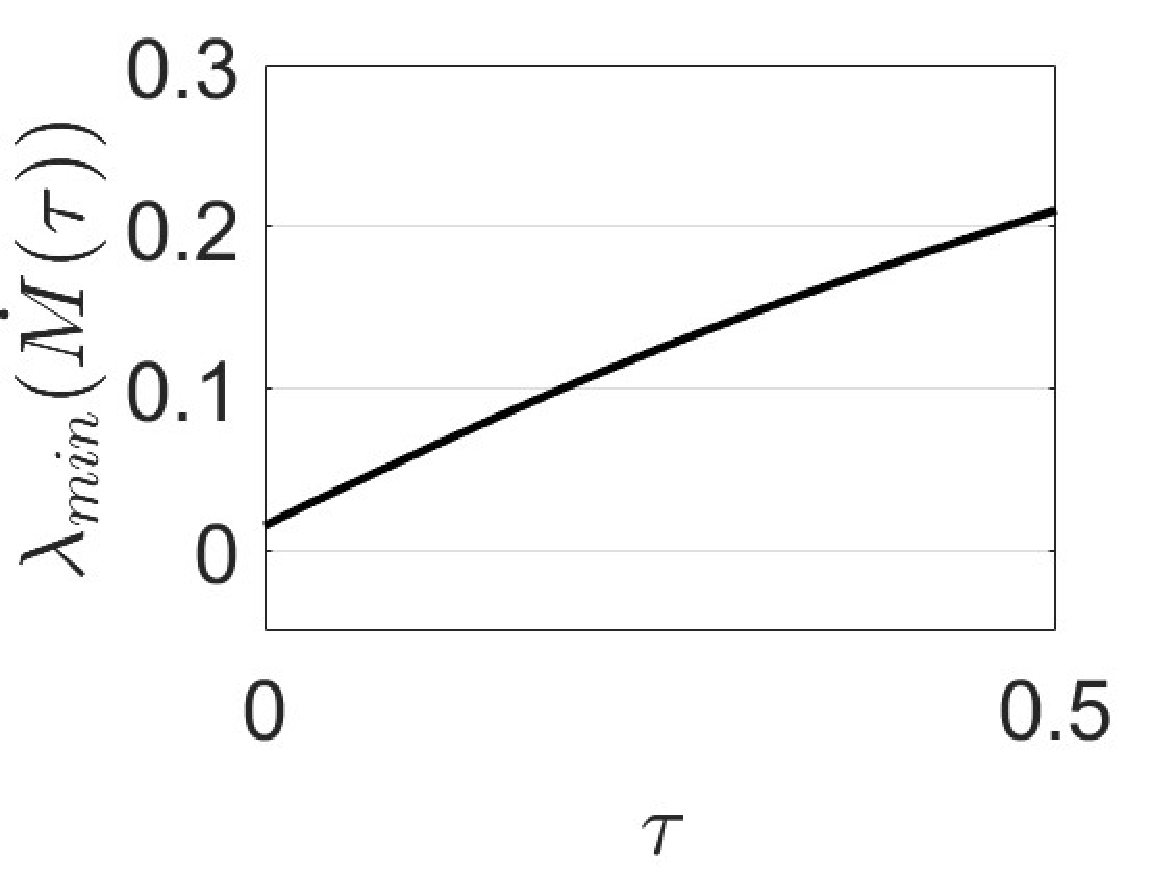}
  	\caption{Evolution of $\lambda_{\min}(\dot{M})$}
  	\label{fig:case1_LmdM}
  \end{subfigure}
  \quad
  \begin{subfigure}{4cm}
    \includegraphics[width=\textwidth]{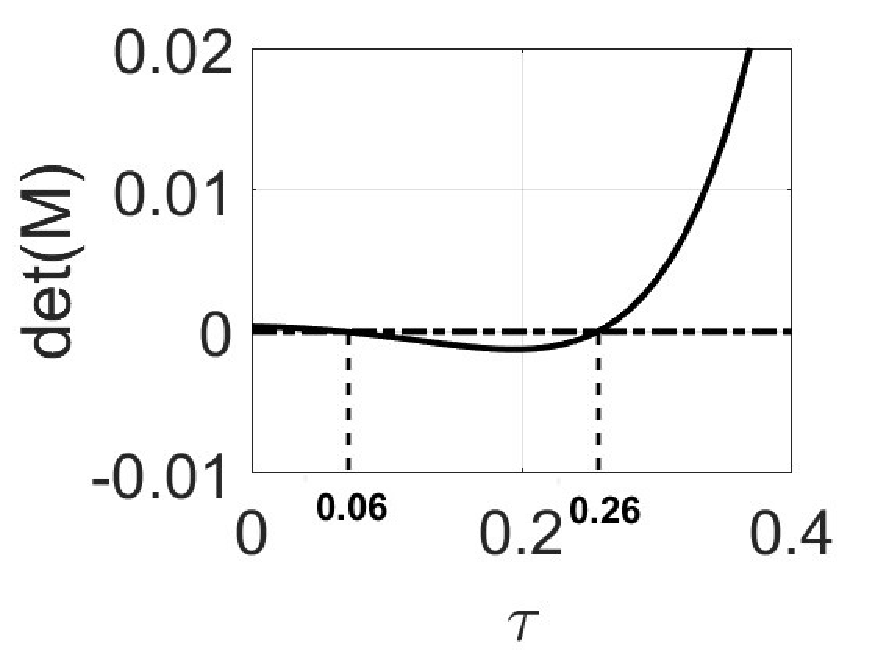}
    \caption{Evolution of $\det(M)$}
    \label{fig:case1_detM}
  \end{subfigure}
  \quad
  \begin{subfigure}{4cm}
    \includegraphics[width=\textwidth]{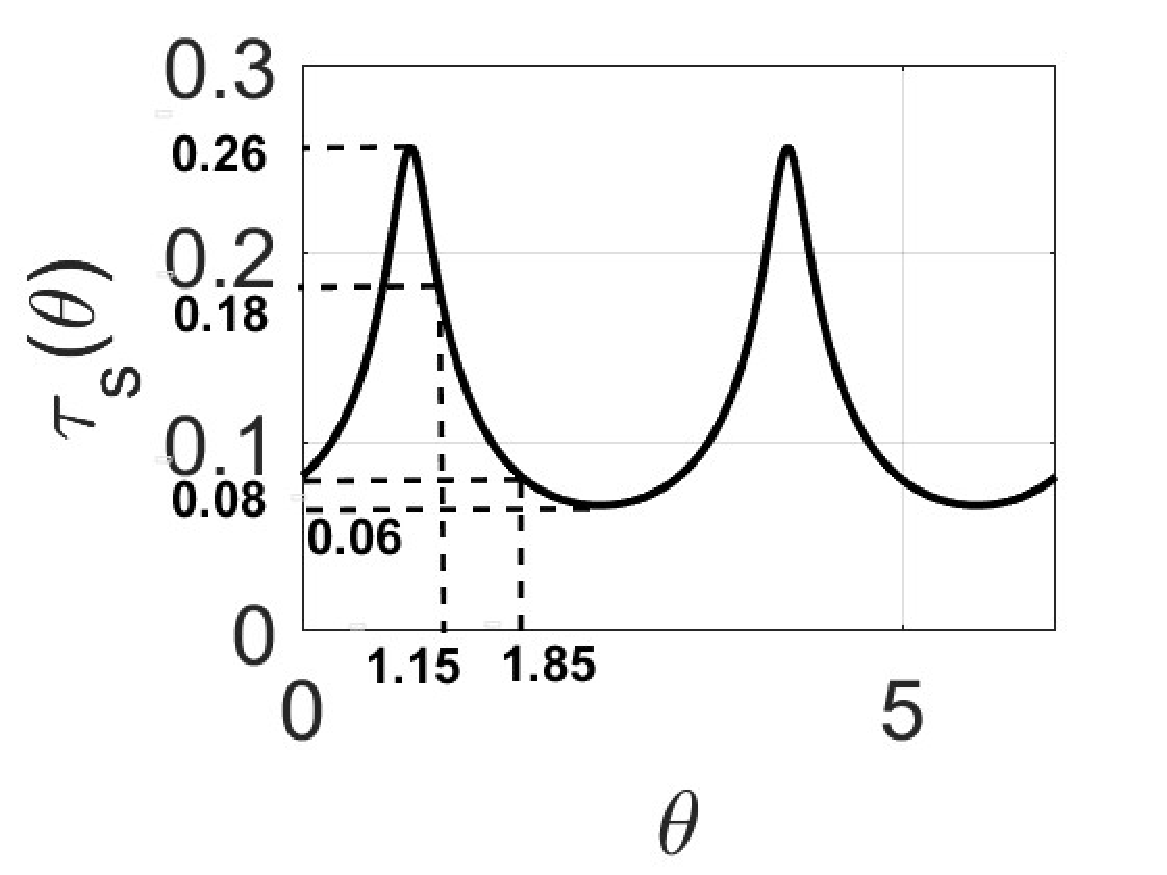}
    \caption{Inter-event time function}
    \label{fig:case1_tau}
  \end{subfigure}
  \quad
  \begin{subfigure}{4cm}
    \includegraphics[width=4cm]{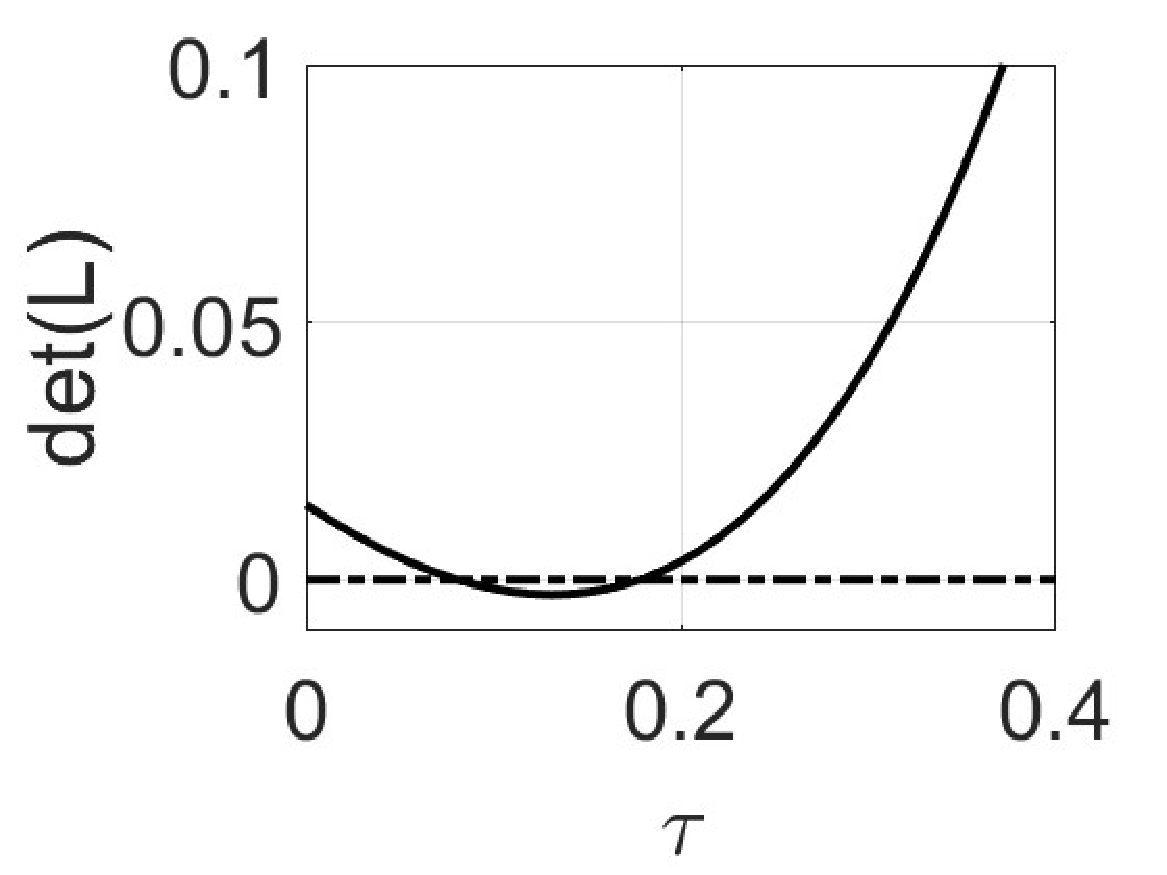}
    \caption{Evolution of $\det(L)$}
    \label{fig:case1_detL}
  \end{subfigure}
  \quad
  \begin{subfigure}{4cm}
    \includegraphics[width=4cm]{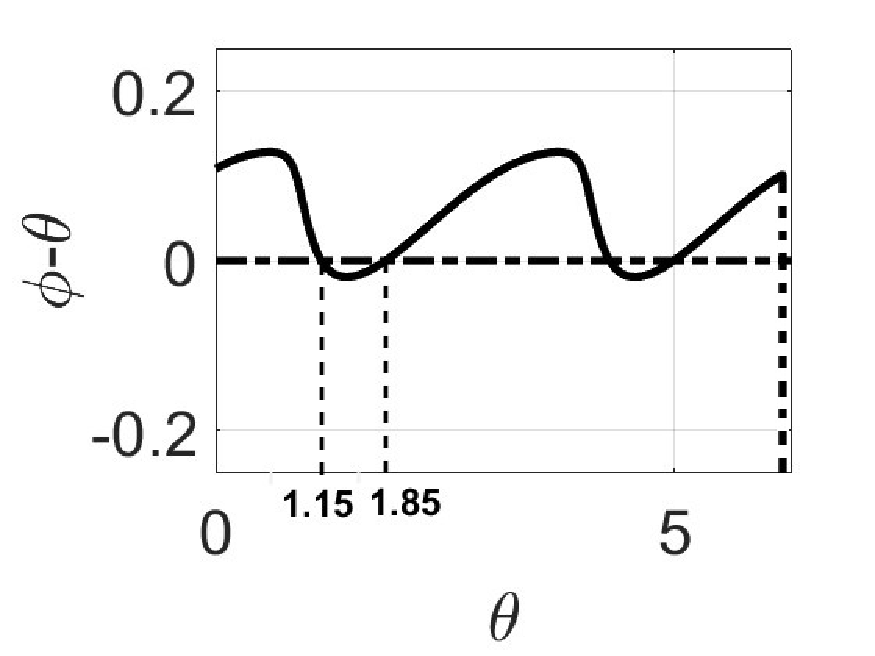}
    \caption{$\phi(\theta)-\theta$ map}
    \label{fig:case1_angle_map}
  \end{subfigure}
  \quad
  \begin{subfigure}{4cm}
  	\includegraphics[width=4cm]{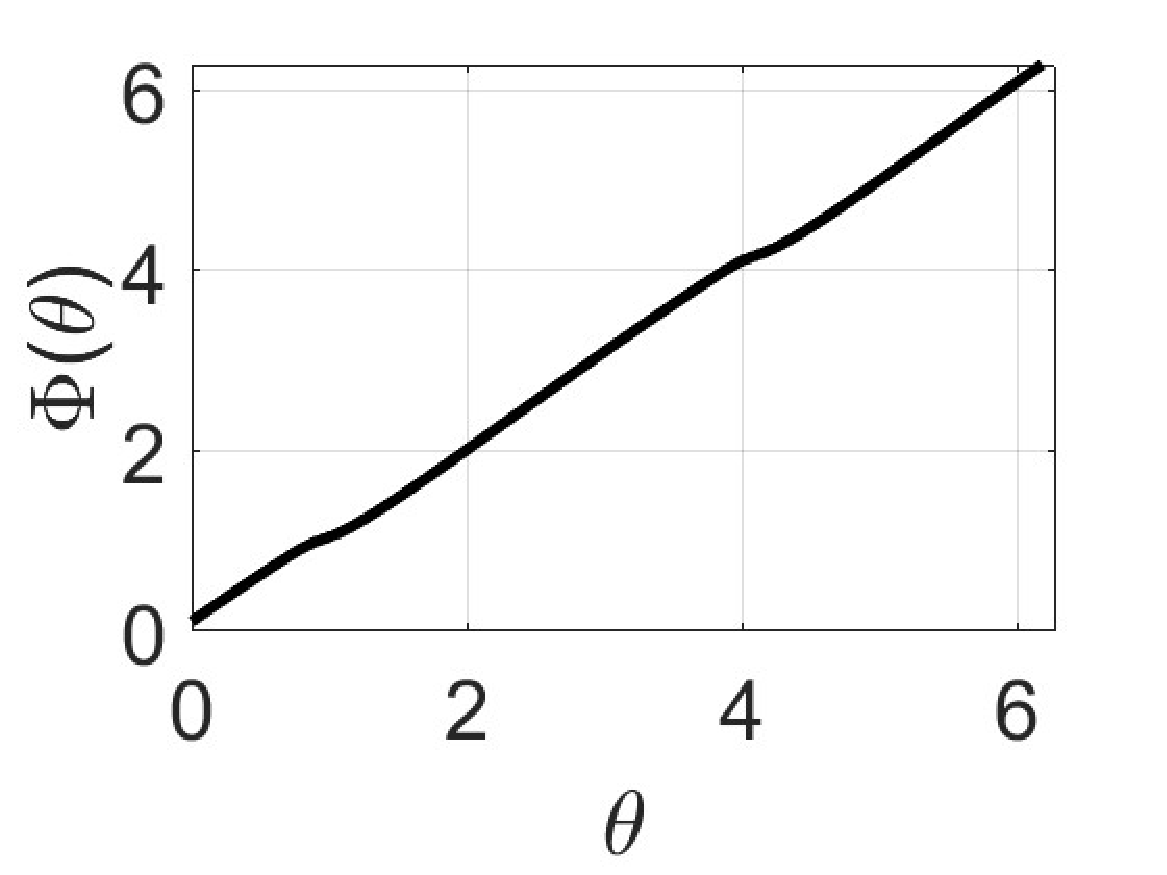}
  	\caption{Lift of the \emph{angle map}}
  	\label{fig:Lift_1}
  \end{subfigure}
  \quad
  \begin{subfigure}{4cm}
    \includegraphics[width=4cm]{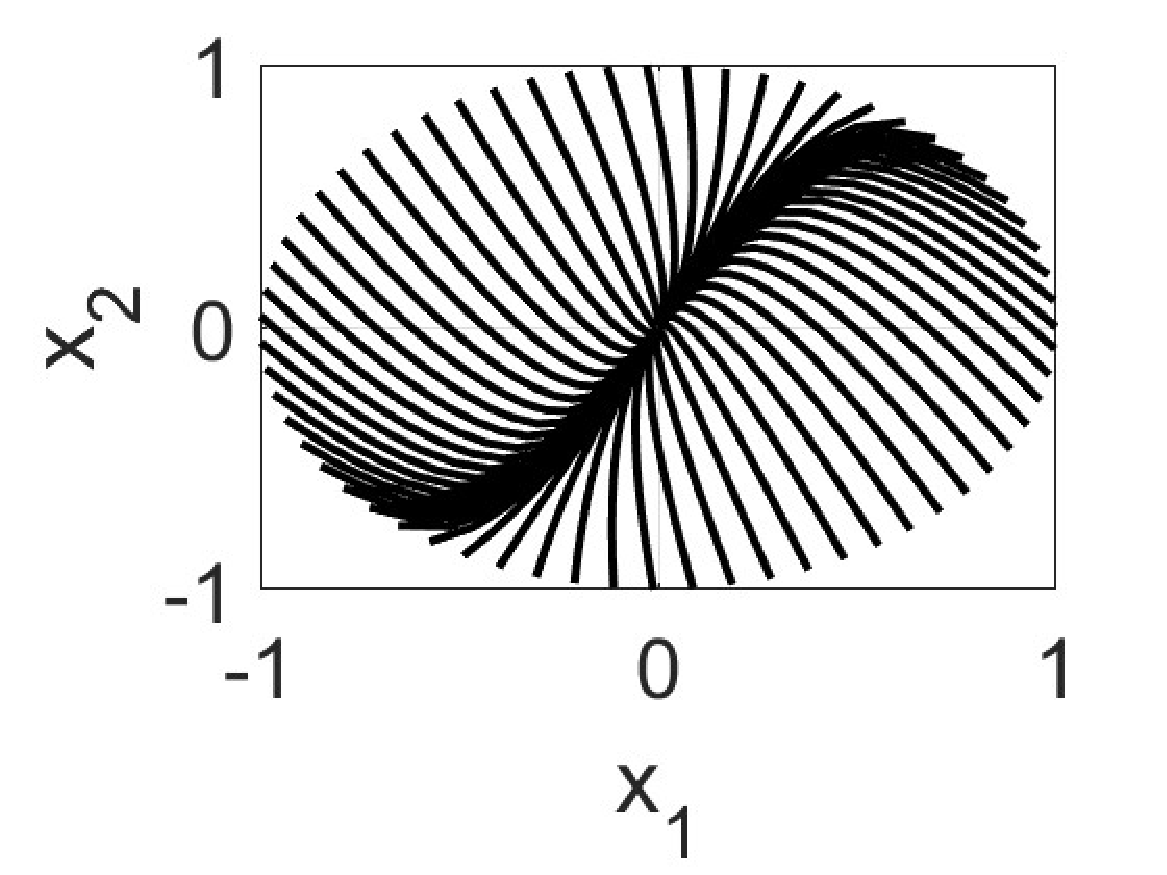}
    \caption{Phase portrait}
    \label{fig:case1_phase_portrait}
  \end{subfigure}
  \quad
  \begin{subfigure}{4cm}
    \includegraphics[width=4cm]{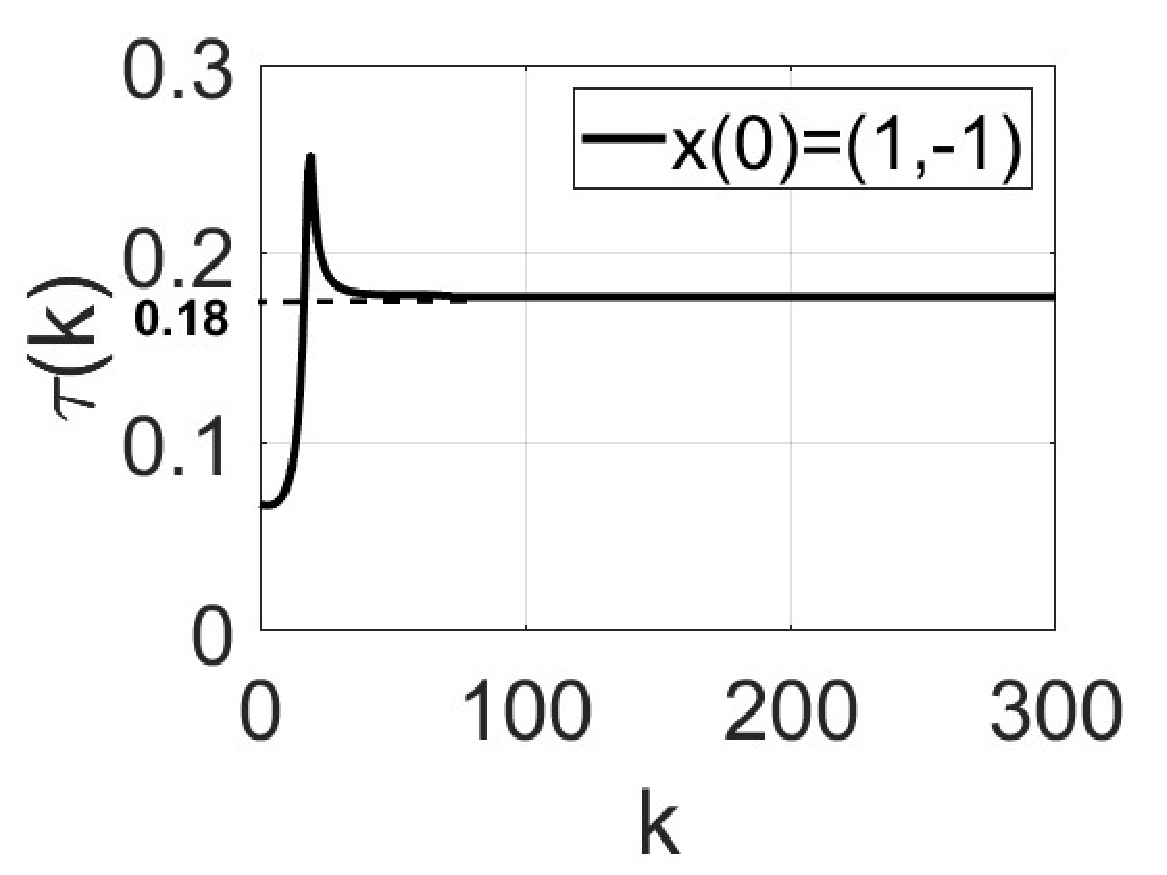}
    \caption{Inter-event time evolution}
    \label{fig:case1_tau_k}
  \end{subfigure}
\quad
\begin{subfigure}{4cm}
	\includegraphics[width=4cm]{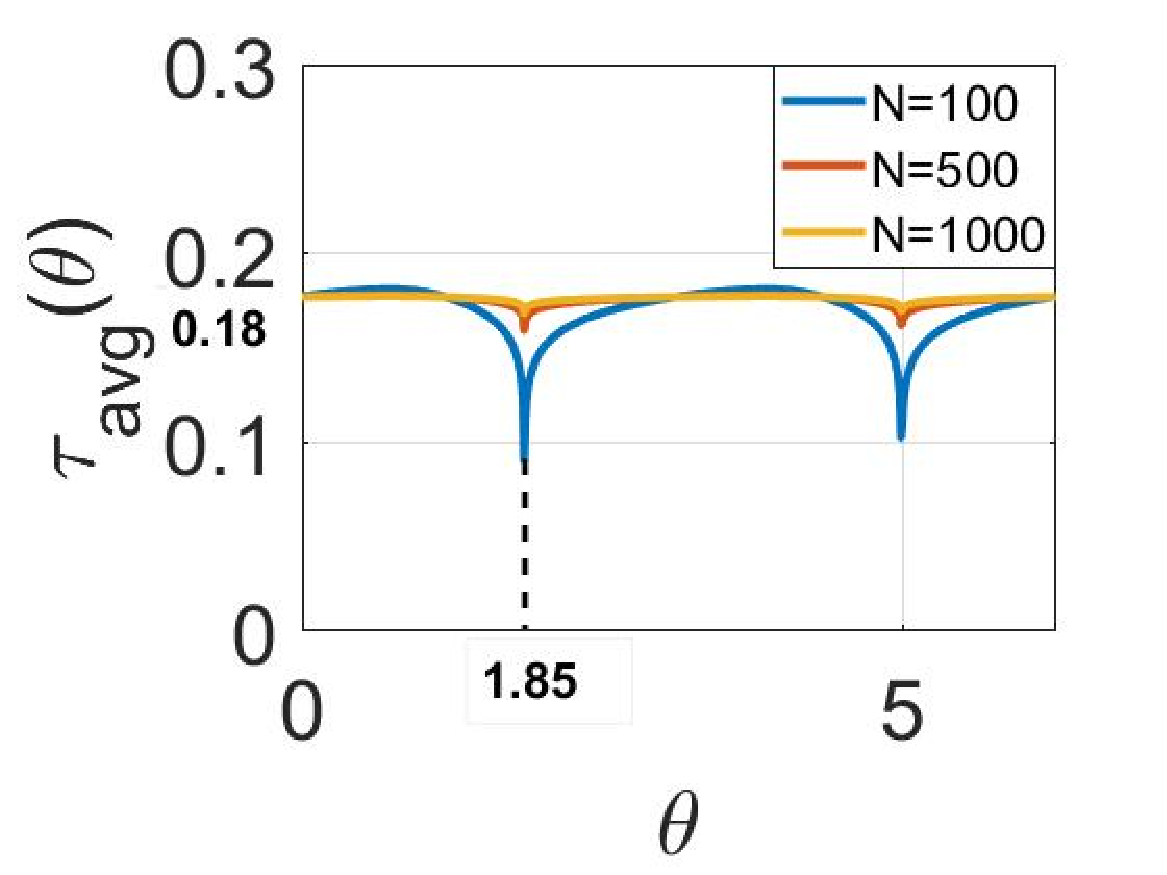}
	\caption{Average inter-event time}
	\label{fig:tau_avg1}
\end{subfigure}
  \caption{Simulation results of Case 1 when $A_c$ has real
  eigenvalues at $\{-0.5528,-1.4472\}$.}
  \label{fig:Ac_with_real_evs}
\end{figure}

\subsubsection*{Case 2:}
\hspace{0.1em} In this case, we use the same $A$ matrix as in Case 1
but choose the input matrix $B=[0\quad1]^T$ and the control gain
$K=[0\quad-5]$ so that $A_c$ has complex conjugate eigenvalues at
$\{-1+i, -1-i\}$. Figure~\ref{fig:Ac_with_complex_evs} shows the
simulation results for Case 2. For this case also the inter-event time
function is continuously differentiable and periodic with period $\pi$. From
Figure~\ref{fig:case2_detM} and Figure~\ref{fig:case2_tau} we can
verify that $\det(M(\tau))=0$ has exactly two solutions and these two
points are $\tmin$ and $\tmax$
respectively. Figure~\ref{fig:case2_detL} shows that $\det(L(\tau))$
is always positive. Therefore the $\phi$ map in
Figure~\ref{fig:case2_angle_map} has no fixed point. Figure~\ref{fig:case2_angle_map} shows that the sufficient condition, given in Theorem~\ref{thm:IET-analysis}, for non-convergence of inter-event times to a steady-state value is satisfied. Hence, we
  can say that the inter-event times do not converge to a steady
  state value for any initial condition. Note that, under the 
  event-triggering rule~\eqref{eq:tr2}, we can also use 
  Remark~\ref{rem:RT_non-convergence} to show the non-convergence of 
  inter-event times to a steady-state value.
Figure~\ref{fig:Lift_2}
presents a lift of the \emph{angle map} $\phi(.)$. As the lift
$\Phi(.)$ is increasing monotonically, the \emph{angle map} $\phi(.)$ is
an orientation preserving homeomorphism.
Figure~\ref{fig:case2_phase_portrait} represents the phase portrait of
the closed loop system. Figure~\ref{fig:case2_tau_k} shows the
evolution of inter-event times, for two arbitrary initial conditions,
is oscillating in nature. Figure~\ref{fig:tau_avg2} presents the
average inter-event time, for the total number of sampling instants
equal to 1000, as a function of the angle of the initial state of
the system. For this case, the average inter-event time is a constant
for all inital states of the system. This may be due to several
reasons. Either there does not exist a periodic orbit for the \emph{angle map} or the \emph{angle map} has a unique periodic orbit with period greater
than one or the average inter-event time correspoding to all periodic
orbits of the \emph{angle map} is the same. It is not easy to distinguish
between these cases from the simulation results.

\begin{figure}[ht]
  \centering
   \begin{subfigure}{4cm}
  	\includegraphics[width=\textwidth]{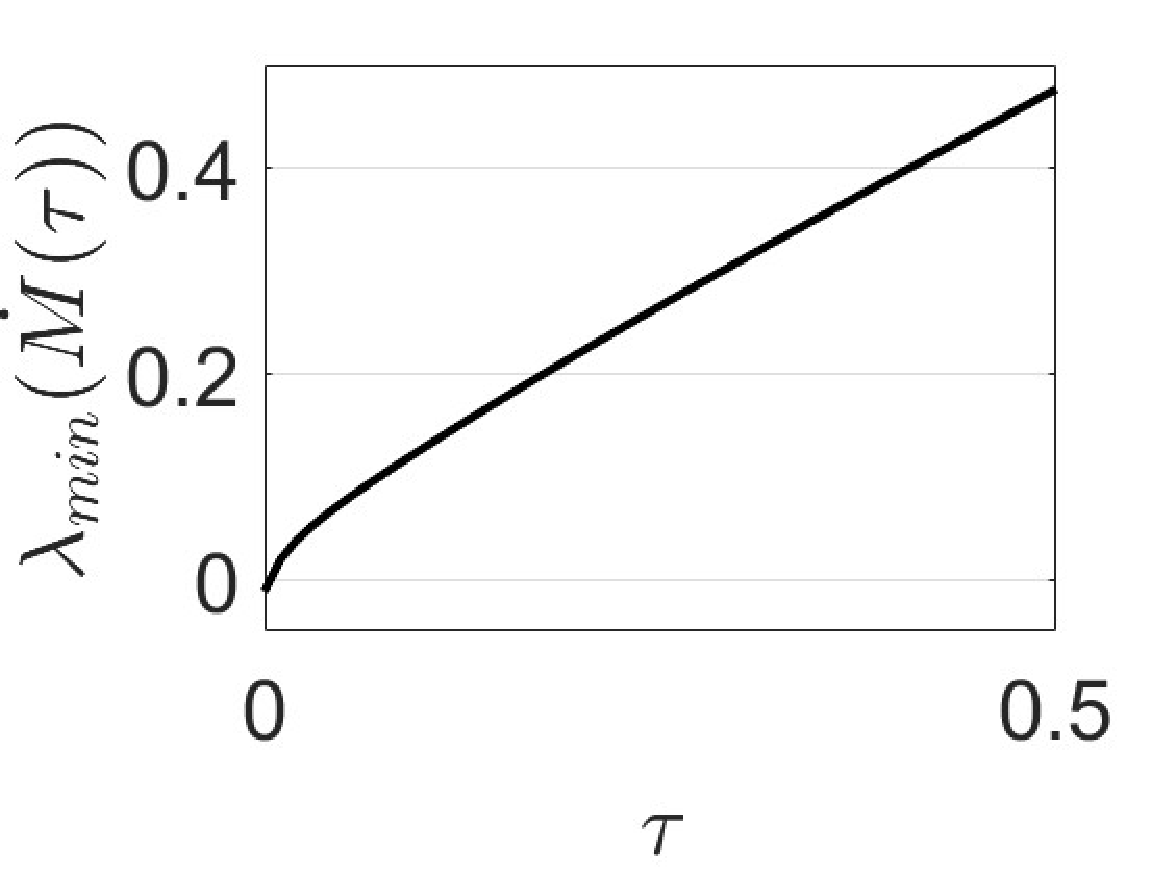}
  	\caption{Evolution of $\lambda_{\min}(\dot{M})$}
  	\label{fig:case2_LmdM}
  \end{subfigure}
  \quad
  \begin{subfigure}{4cm}
    \includegraphics[width=4cm]{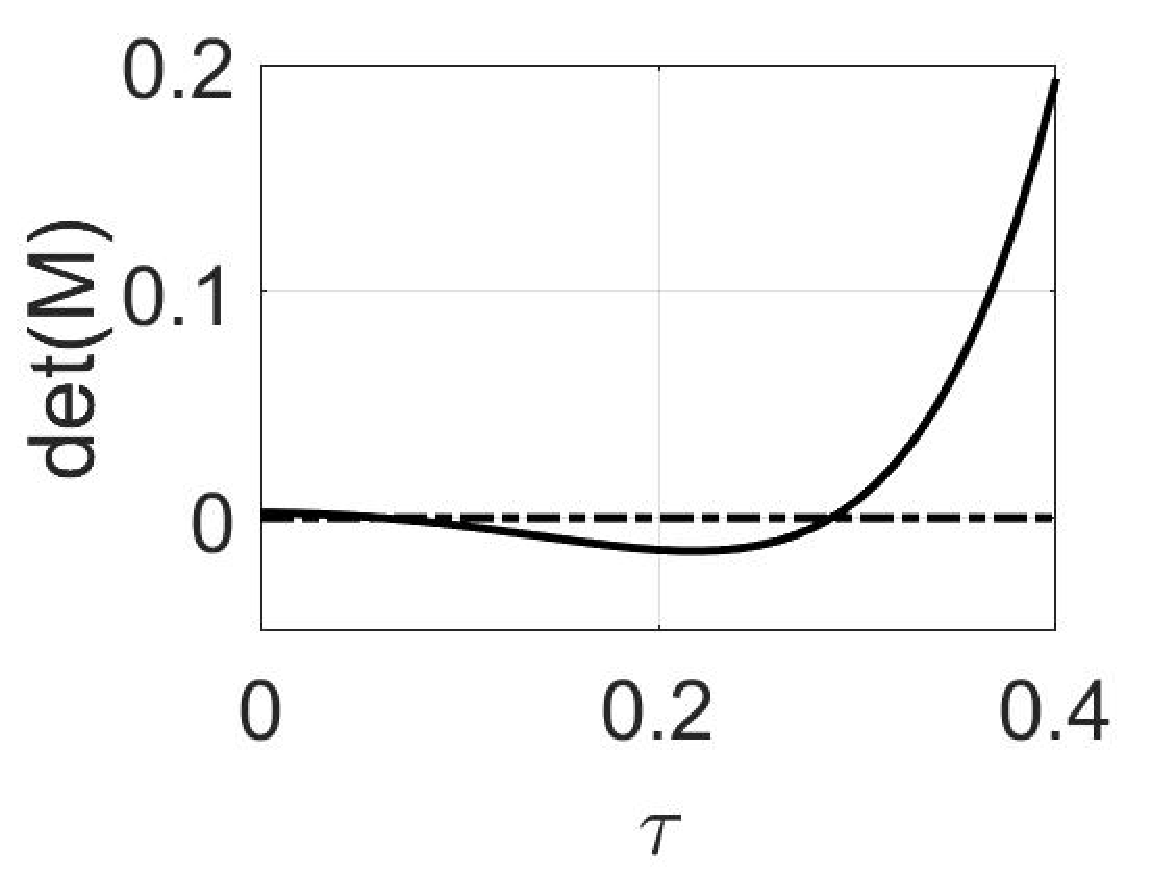}
    \caption{Evolution of $\det(M)$}
    \label{fig:case2_detM}
  \end{subfigure}
  \quad
  \begin{subfigure}{4cm}
    \includegraphics[width=4cm]{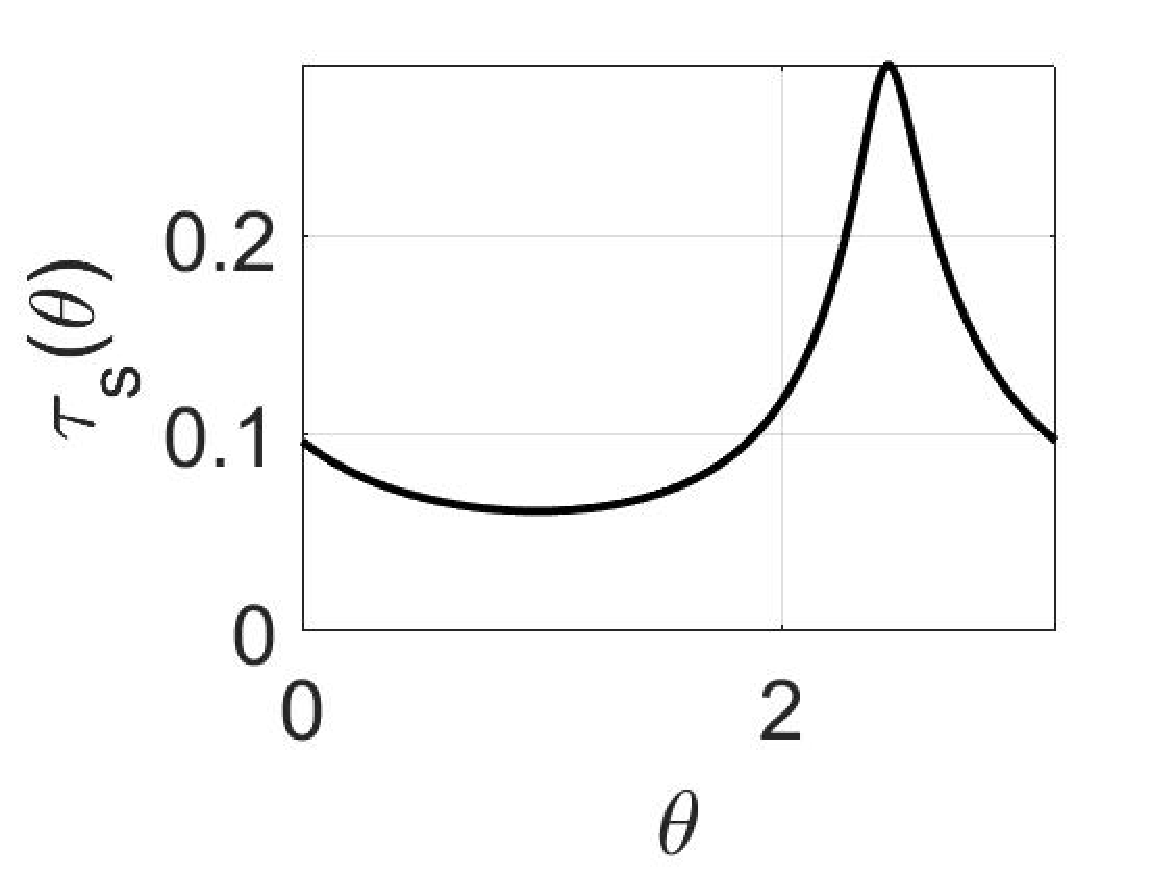}
    \caption{Inter-event time function}
    \label{fig:case2_tau}
  \end{subfigure}
  \quad
  \begin{subfigure}{4cm}
    \includegraphics[width=4cm]{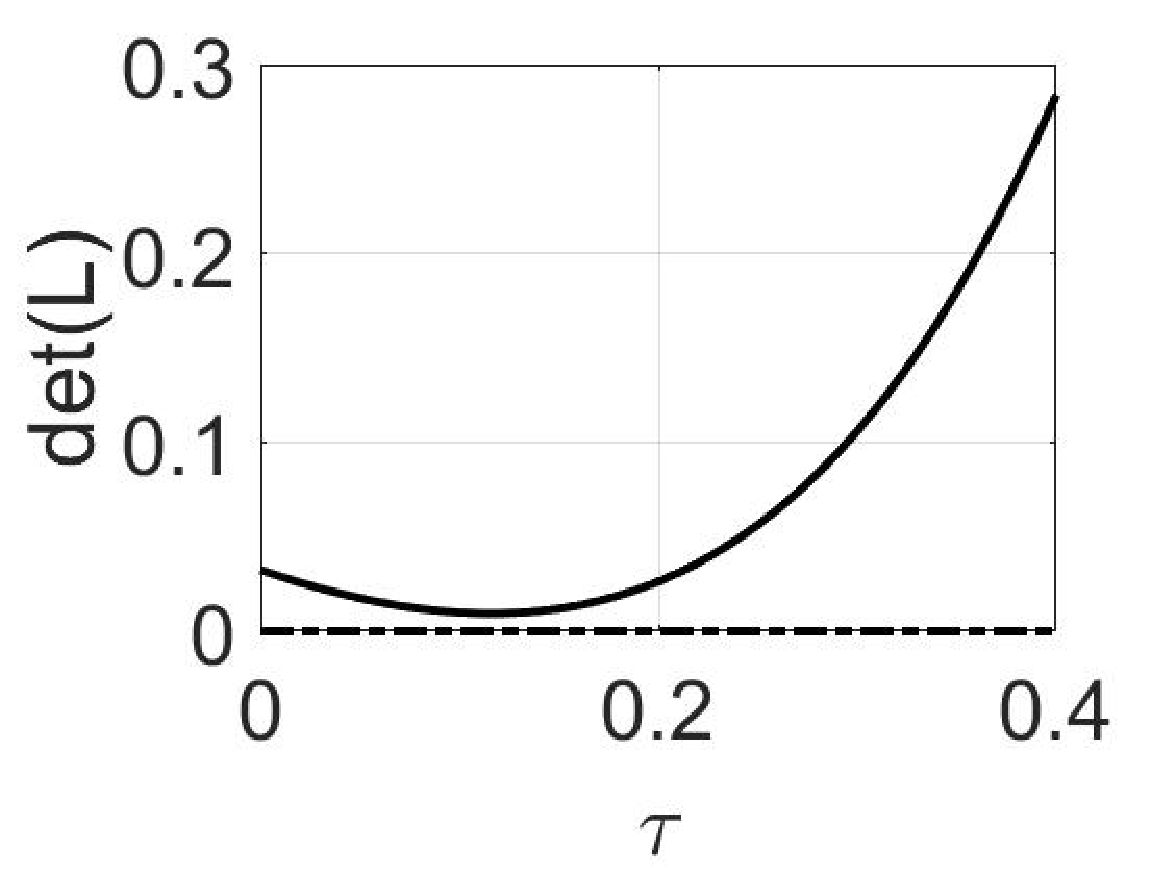}
    \caption{Evolution of $\det(L)$}
    \label{fig:case2_detL}
  \end{subfigure}
  \quad
  \begin{subfigure}{4cm}
    \includegraphics[width=4cm]{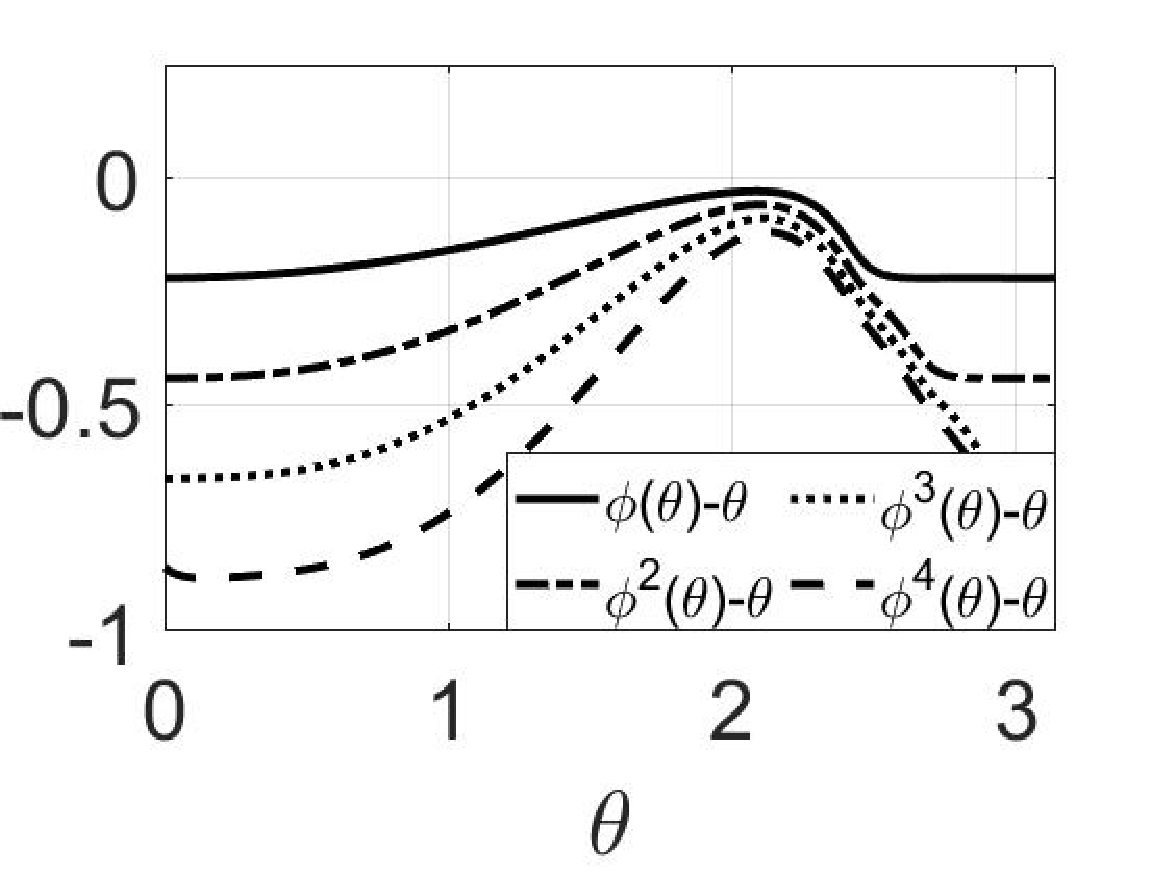}
    \caption{$\phi^k(\theta)-\theta$ map}
    \label{fig:case2_angle_map}
  \end{subfigure}
\quad
\begin{subfigure}{4cm}
	\includegraphics[width=4cm]{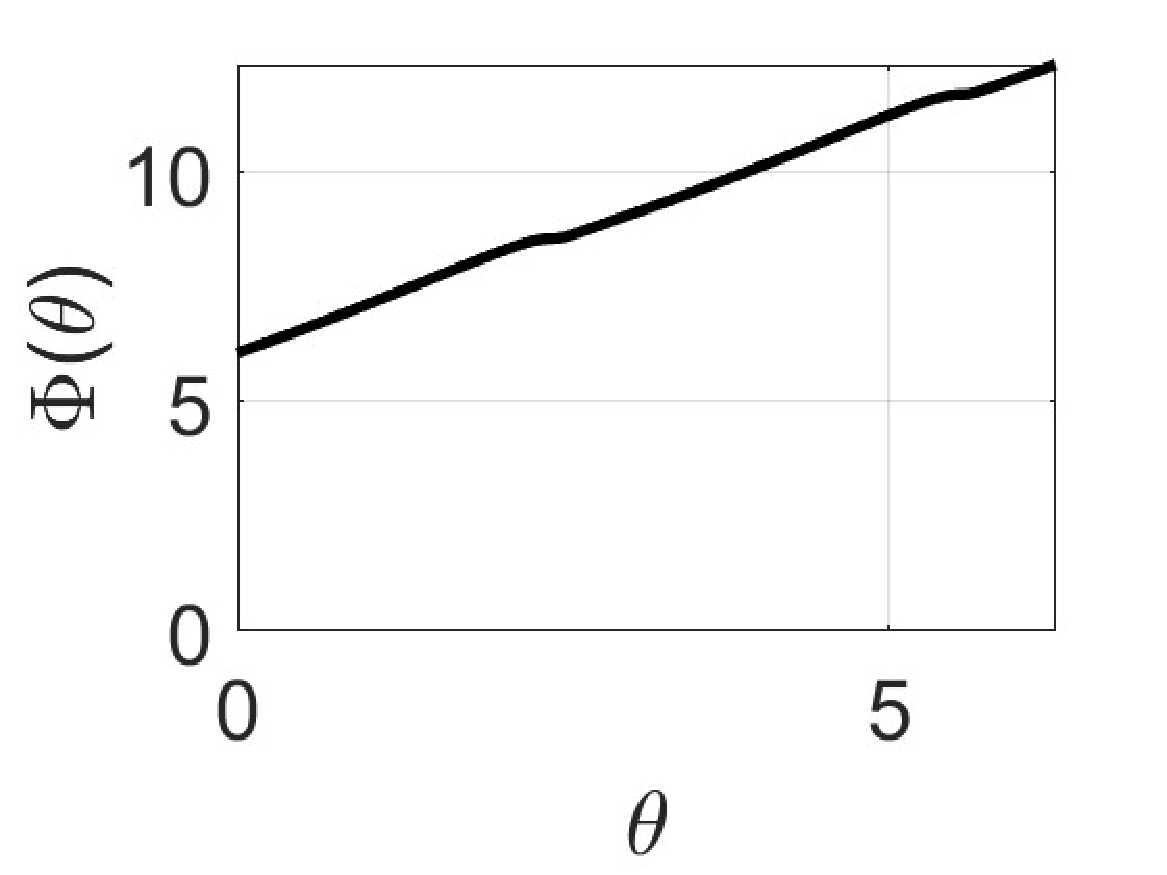}
	\caption{Lift of the \emph{angle map}}
	\label{fig:Lift_2}
\end{subfigure}
  \quad
  \begin{subfigure}{4cm}
    \includegraphics[width=4cm]{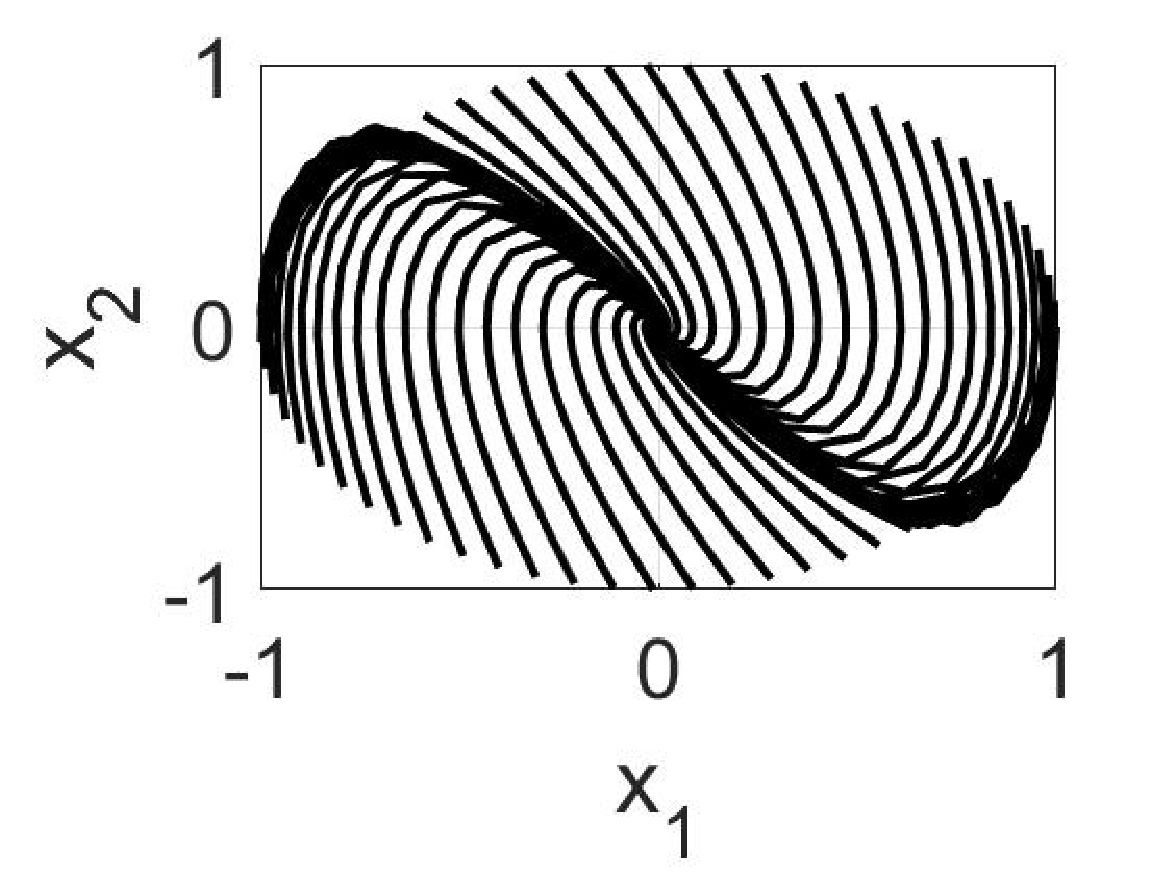}
    \caption{Phase portrait of the closed loop system}
    \label{fig:case2_phase_portrait}
  \end{subfigure}
  \quad
  \begin{subfigure}{4cm}
    \includegraphics[width=4cm]{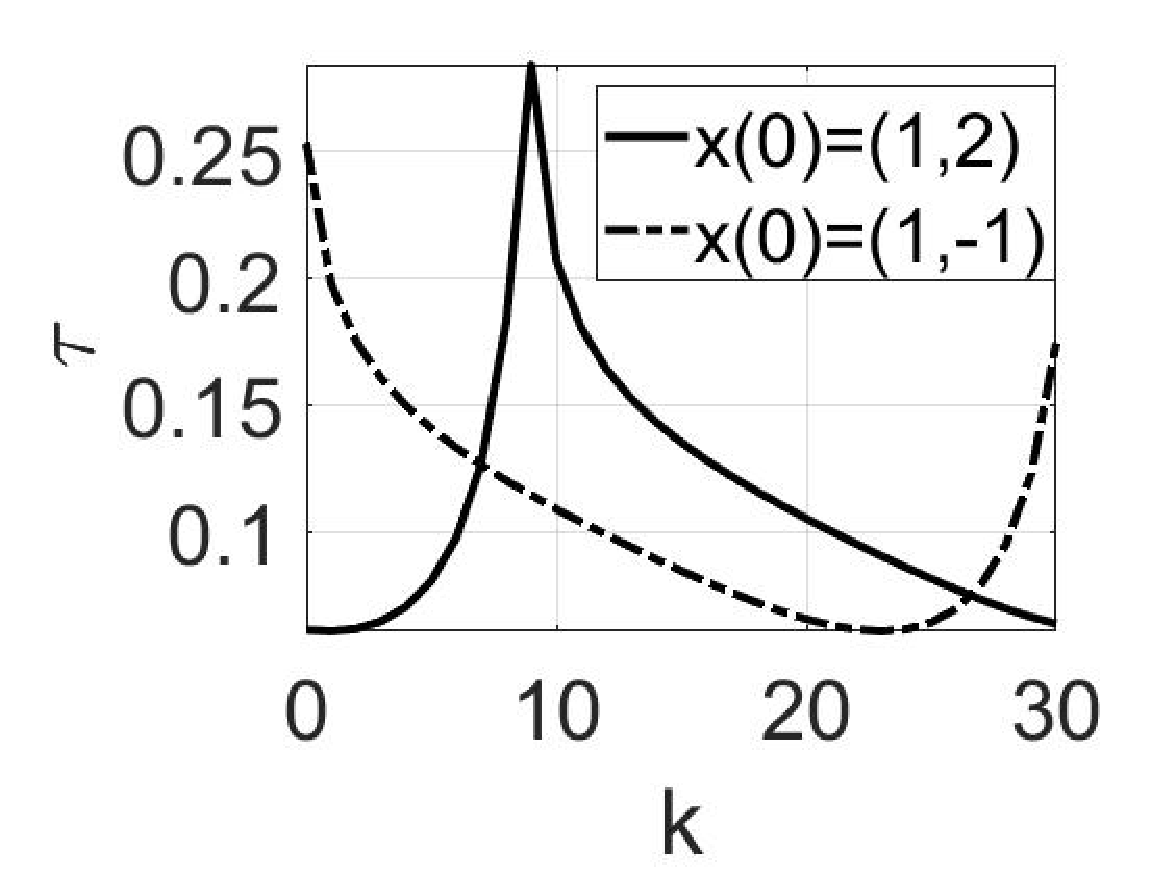}
    \caption{Evolution of inter-event times}
    \label{fig:case2_tau_k}
  \end{subfigure}
\quad
\begin{subfigure}{4cm}
	\includegraphics[width=4cm]{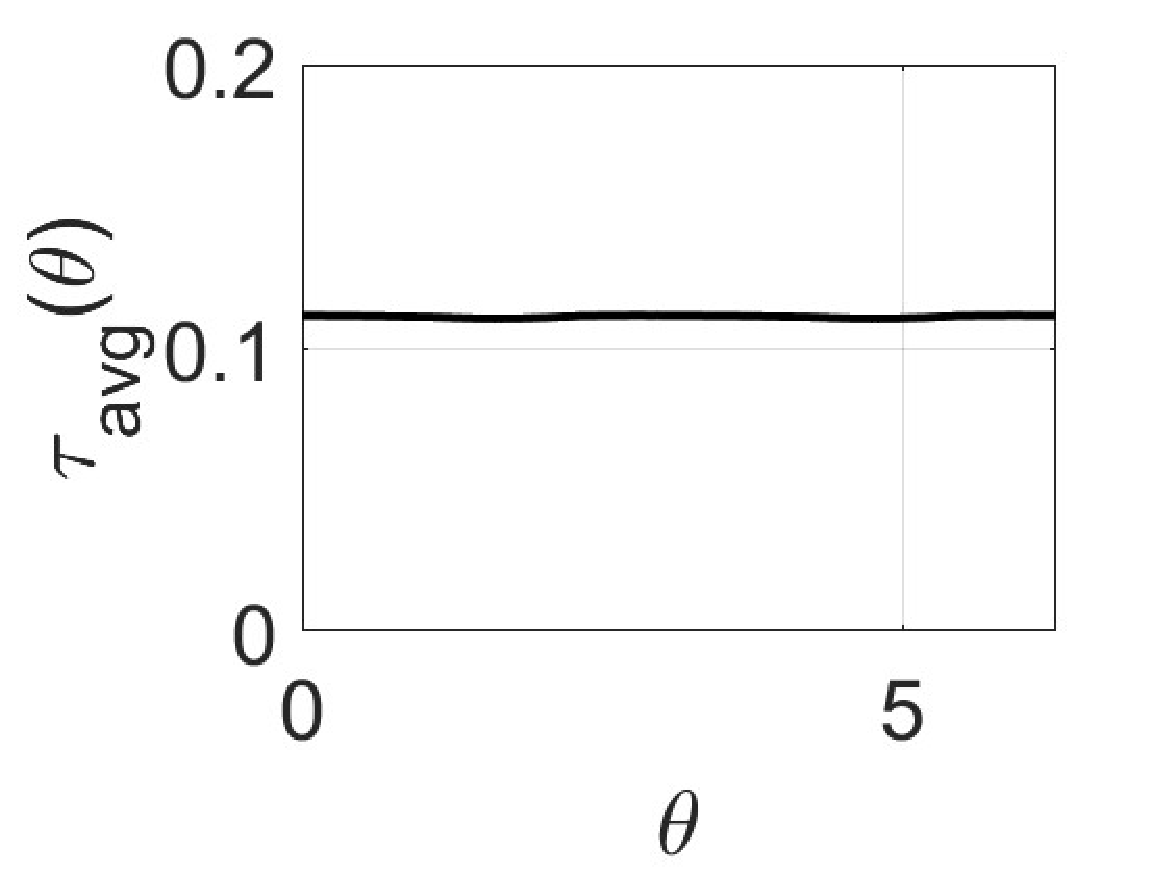}
	\caption{Average inter-event time}
	\label{fig:tau_avg2}
\end{subfigure}
  \caption{Simulation results of Case 2 when $A_c$ has complex
  conjugate eigenvalues at $\{-1+i,-1-i\}$.}
  \label{fig:Ac_with_complex_evs}
\end{figure}

\subsubsection*{Case 3:}
Consider another system,
\begin{equation*}
  \dot{x}=\begin{bmatrix*}[r]
    0 & 1 \\ -2 & 3
    \end{bmatrix*}x+\begin{bmatrix} 1 & 0 \\ 0 & 1
  \end{bmatrix}u \rdef Ax + Bu.
\end{equation*}
The system matrix $A$ has real eigenvalues at $\{1,2\}$. The control
gain $K=\begin{bmatrix*}[r] -1 & -0.8 \\ 1.8 & -4
\end{bmatrix*}$ so that $A_c$ has complex conjugate eigenvalues at
$\{-1+0.2i,-1-0.2i\}$. Figure~\ref{fig:Ac_with_complex_evs_and_fixed_points}
shows the simulation results of this system for the event triggering
rule \eqref{eq:tr2} with
$\sigma=0.2251$. Figure~\ref{fig:case3_angle_map} shows that the \emph{angle map} $\phi(.)$ has two fixed points, where the larger one is a stable
fixed point. Note that according to
  Proposition~\ref{prop:nec-cond-fixed-point}, there exists a fixed
  point for the \emph{angle map} only if $\left\lVert R \right\rVert >1$. In
  this case, we can verify that
  $\left\lVert R \right\rVert = 1.3136$. In
Figure~\ref{fig:case3_tau_k} the inter-event time is converging to a
steady state value for two different initial conditions. Under
  the assumption of sufficiently small relative threshold parameters,
  ~\cite{RP-RS-WH-2022} says that if the eigenvalues of the closed
  loop system matrix $A_c$ are complex conjugates then the inter-event
  times oscillate in a near periodic manner. But, in this example we
show that even if $A_c$ has only complex conjugate eigenvalues, the
inter-event times may still converge to a steady state
value. Note however, that we cannot claim that this is a
  counter-example to the results of~\cite{RP-RS-WH-2022} as the bound
  on the relative thresholding parameter for which their results hold
  is not explicitly stated.

\begin{figure}[ht]
  \centering
   \begin{subfigure}{4cm}
    \includegraphics[width=4cm]{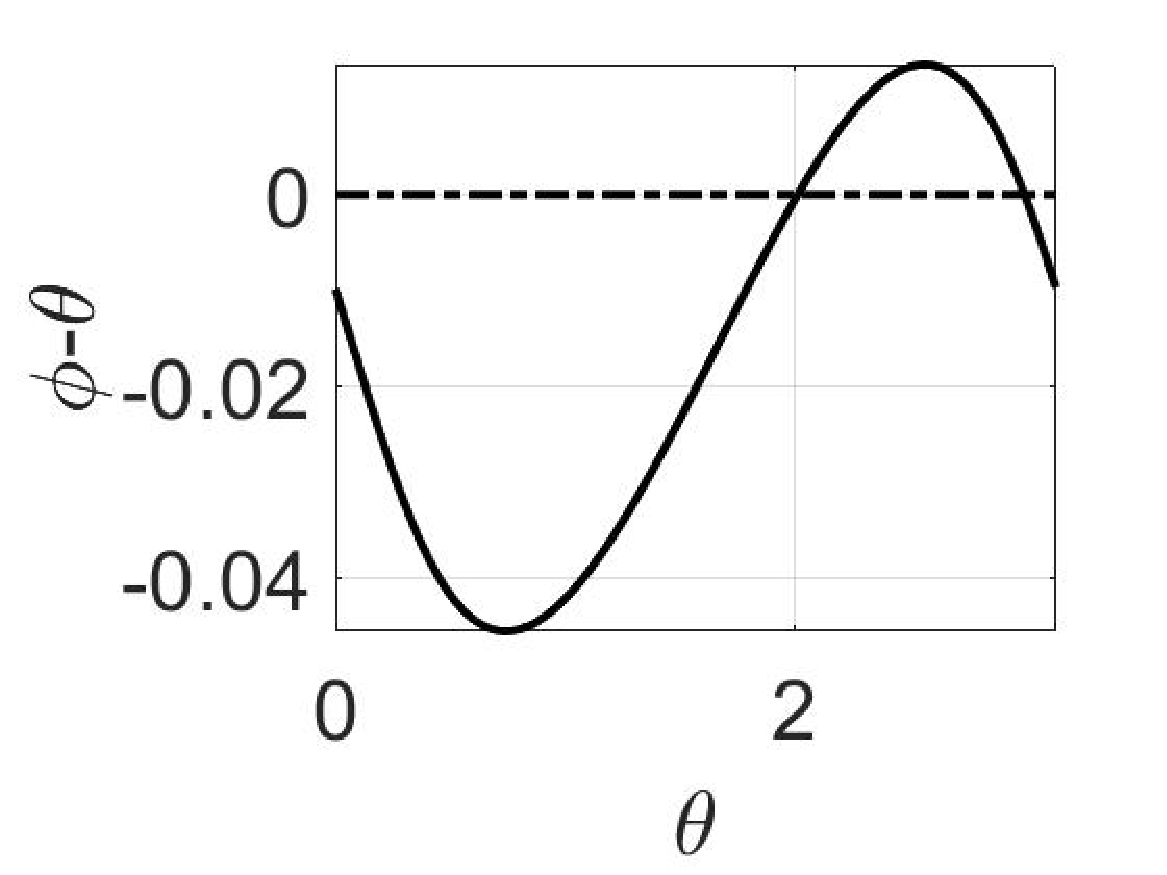}
    \caption{$\phi(\theta)-\theta$ map}
    \label{fig:case3_angle_map}
  \end{subfigure}
  \quad
  \begin{subfigure}{4cm}
    \includegraphics[width=4cm]{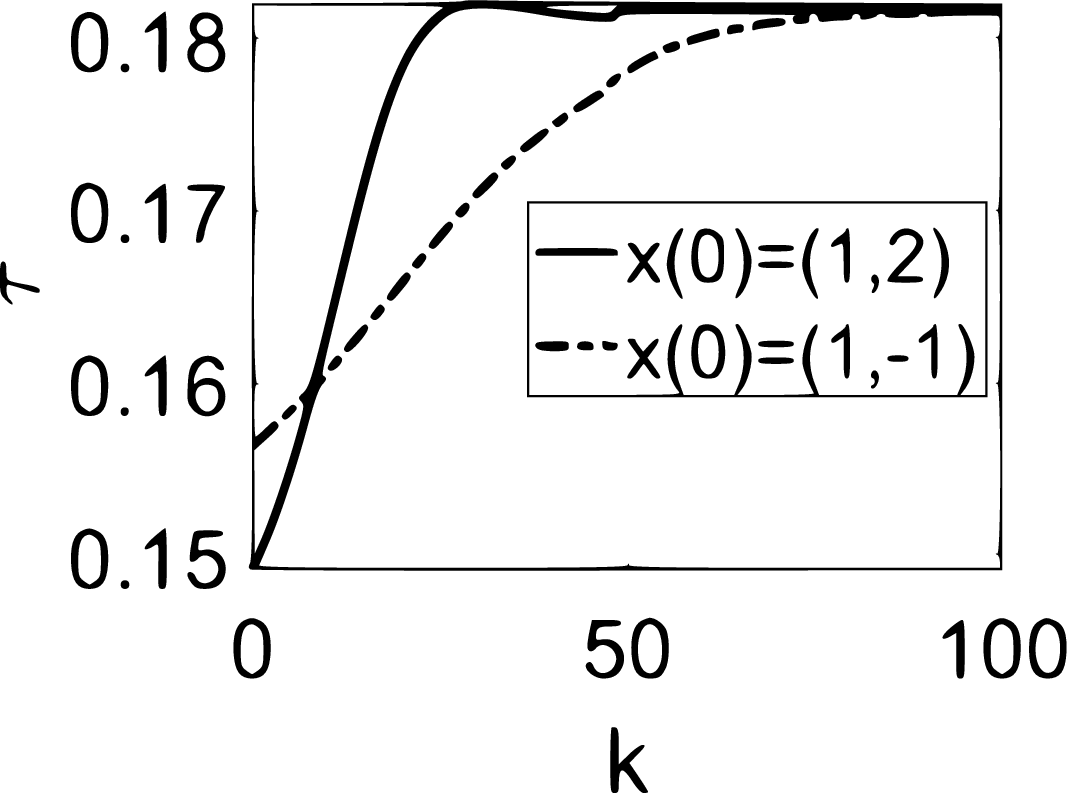}
    \caption{Evolution of inter-event times}
    \label{fig:case3_tau_k}
  \end{subfigure}
  \caption{Simulation results of Case 3 when $A_c$ has complex
  conjugate eigenvalues at $[-1+0.2i,-1-0.2i]$.}
  \label{fig:Ac_with_complex_evs_and_fixed_points}
\end{figure}

\subsubsection*{Case 4:}
Now consider the system,
\begin{equation*}
  \dot{x}=\begin{bmatrix}
    1 & 4 \\ 0 & 1
    \end{bmatrix}x+\begin{bmatrix} 0 \\ 1
  \end{bmatrix}u \rdef Ax + Bu .
\end{equation*}
$A$ has real and equal eigenvalues at $\{1,1\}$. We let the control
gain $K=[-2 \quad -4]$, so that $A_c$ has eigenvalues at
$\{-1+2i,-1-2i\}$. Figure~\ref{fig:discontinuous_iet} shows the
simulation results of this system for the triggering rule
\eqref{eq:tr1}. Figure~\ref{fig:case4_tau} shows that the inter-event
time function $\tau_s(\theta)$ is discontinuous around $\theta=$ 2.3
radians. In Figure~\ref{fig:case4_level_set}, we can see that around
$\theta=2.3$ radians, there is a jump in the smallest $\tau$ value at
which $f_s(\theta,\tau)=0$. This causes a point of discontinuity in
the inter-event time function.

\begin{figure}[ht]
  \centering
  \begin{subfigure}{4cm}
    \includegraphics[width=4cm]{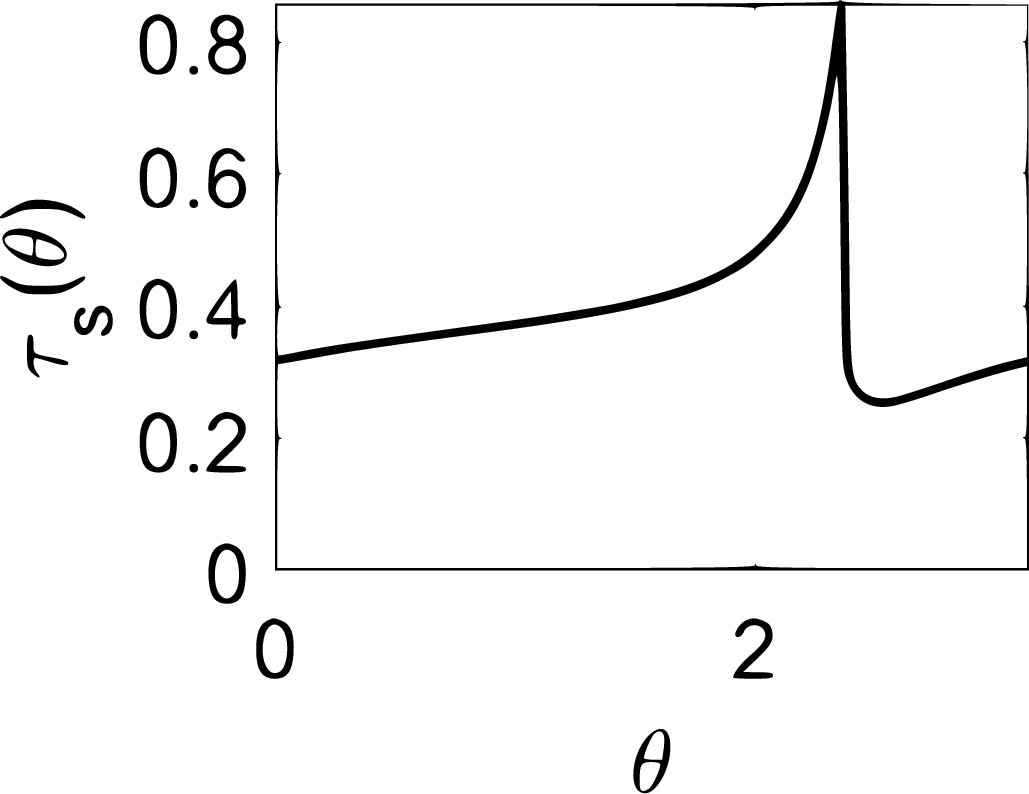}
    \caption{Inter-event time function}
    \label{fig:case4_tau}
  \end{subfigure}
  \quad
  \begin{subfigure}{4cm}
    \includegraphics[width=4cm]{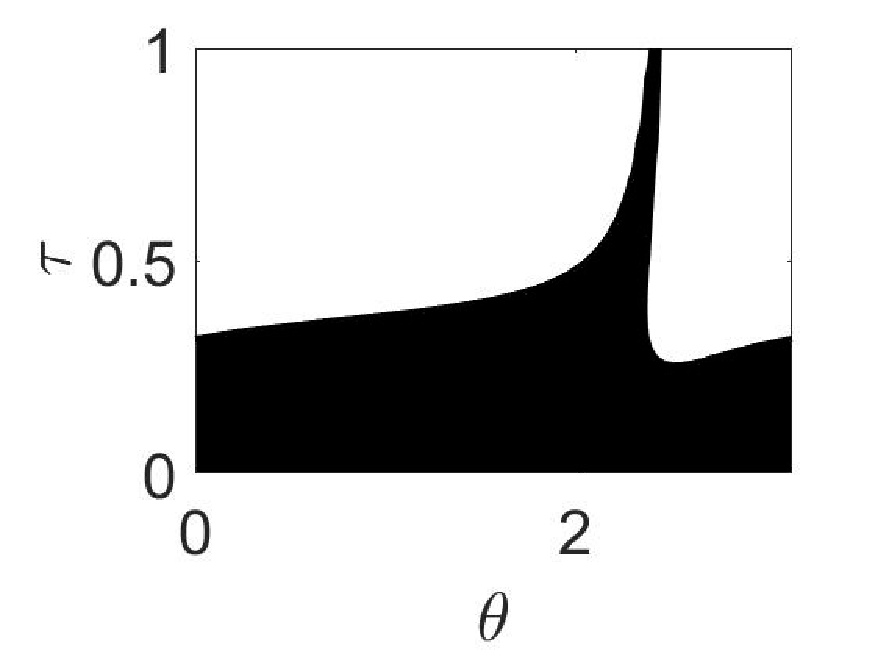}
    \caption{0-sub-level set of $f_s$}
    \label{fig:case4_level_set}
  \end{subfigure}
  \caption{Simulation results of Case 4 with discontinuous inter-event
  time function.}
  \label{fig:discontinuous_iet}

\end{figure}

\section{Conclusion} \label{sec:conclusion}

In this paper, we analyzed the asymptotic behavior of inter-event
times in planar linear systems under a general class of
scale-invariant event-triggering rules. As the inter-event time is 
a function of the angle of the state at an event, we carried out 
inter-event time analysis indirectly by studying the evolution of the 
angle of the state from one event to the next. The analysis of 
evolution of inter-event times is complex even for planar systems and 
the results in this paper are among very few in the literature that 
even seek to explain the variety of evolutions that is possible for the 
inter-event times. The proposed analytical results on the evolution of 
inter-event times are not directly
	extendable to a general n-dimensional system. However, the idea that analyzing the state evolution from one event
	to next as a means to analyzing the evolution of inter-event times
	does certainly apply to n-dimensional systems. We have used the same 
	idea in one of our recent work~\cite{AR-PT:2023} to analyze the 
	evolution of inter-event times of general n-dimensional LTI systems 
	under the region-based self-triggered control method. Future work 
	includes analysis of asymptotic behaviors of inter-event times using 
	an approximate rotation number of the angle map. Another direction 
	could be to determine an approximate asymptotic average inter-event 
	time with known error bounds. One could use more ideas from ergodic 
	theory to do the same. Other potential research directions include 
	extensions of the analysis to event-triggered control systems of 
	higher dimensions and to nonlinear systems, at least in a 
	self-triggered control context. 

\section*{Acknowledgements}

This work was partially supported by Science and Engineering
Research Board under grant CRG/2019/005743. Anusree Rajan was supported
by a fellowship grant from the Centre for Networked Intelligence
(a Cisco CSR initiative) of the Indian Institute of Science.

\bibliographystyle{elsarticle-num}
\bibliography{references}

  \end{document}